\documentclass[a4paper,12pt]{amsart} 

\usepackage[lmargin=2.4cm,rmargin=2.25cm,bmargin=3cm,tmargin=2.5cm]{geometry}

\usepackage[utf8]{inputenc}
\usepackage{amsmath}
\usepackage{amssymb}  
\usepackage{amsthm}
\usepackage{textcomp}
\usepackage{enumerate}
\usepackage[round,sort]{natbib}
\usepackage{array}
\usepackage{bbm}
\usepackage{verbatim}
\usepackage{hyperref}
\usepackage{graphicx}
\usepackage{booktabs}
\usepackage{algorithm}
\usepackage{algpseudocode}

\usepackage{xcolor}
\hypersetup{
  colorlinks,
  linkcolor={red!50!black},
  citecolor={blue!50!black},
  urlcolor={blue!80!black}
}

\newcommand{\defeq}{\mathrel{\mathop:}=} 
\newcommand{\eqdef}{=\mathrel{\mathop:}}

\newcommand{\F}{\mathcal{F}}
\newcommand{\G}{\mathcal{G}}

\newcommand{\sumsymb}{\Sigma}
\newcommand{\iid}{\mathrm{iid}}
\newcommand{\str}{\mathrm{str}}
\newcommand{\res}{\mathrm{res}}
\newcommand{\sys}{\mathrm{sys}}
\newcommand{\ml}{\mathrm{ML}}

\newcommand{\Halmos}{\qedhere}

\newcommand{\uarg}{\,\cdot\,}
\newcommand{\ud}{\mathrm{d}}
\newcommand{\R}{\mathbb{R}}
\newcommand{\N}{\mathbb{N}}
\renewcommand{\P}{\mathbb{P}}
\newcommand{\E}{\mathbb{E}}

\newcommand{\charfun}[1]{\mathbbm{I}\left\{#1\right\}}

\newcommand{\given}{\,:\,}

\newcommand{\var}{\mathrm{var}}
\newcommand{\cov}{\mathrm{cov}}

\newcommand{\bigmid}{\;\big|\;}

\newtheorem{theorem}{Theorem}
\newtheorem{corollary}[theorem]{Corollary}
\newtheorem{proposition}[theorem]{Proposition}
\newtheorem{lemma}[theorem]{Lemma}

\theoremstyle{definition}
\newtheorem{definition}[theorem]{Definition}

\newtheorem{condition}[theorem]{Condition}
\newtheorem{example}[theorem]{Example}

\theoremstyle{remark}
\newtheorem{remark}[theorem]{Remark}

\title{Unbiased estimators and multilevel Monte Carlo}
\author{Matti Vihola}
\address{Matti Vihola, Department of Mathematics and Statistics, University of
Jyväskylä P.O.Box 35, FI-40014 University of Jyväskylä, Finland}
\email{matti.vihola@iki.fi}
\begin{document}

\maketitle
\sloppy
%%%%%%%%%%%%%%%%%%%%
\begin{abstract} %{{{
Multilevel Monte Carlo (MLMC) and unbiased estimators
recently proposed by McLeish ({\em Monte Carlo Methods Appl.}, 2011)
and Rhee and Glynn ({\em Oper. Res.}, 2015) are closely related. This
connection is elaborated by presenting a new general class of unbiased
estimators, which admits previous debiasing schemes as special cases.
New lower variance estimators are proposed, which are stratified
versions of earlier unbiased schemes. Under general conditions,
essentially when MLMC admits the canonical square root Monte Carlo
error rate, the proposed new schemes are shown to be asymptotically as
efficient as MLMC, both in terms of variance and cost. The experiments
demonstrate that the variance reduction provided by the new schemes
can be substantial.  
\end{abstract} %}}}

%%%%%%%%%%%%%%%%%%%%%%%%%%%%%%%%%%%%%%%%%%%%%%%%%%%%%%%%%%%%%%%%%%%%%%%%%%%%%%
\section{Introduction} %{{{

Multilevel Monte Carlo (MLMC) methods pioneered by
\cite{heinrich} and \cite{giles-or} are now standard for
estimation of expectations of functionals of processes defined by
stochastic differential equations (SDEs). While the MLMC techniques
have origins in the integral operators \citep{heinrich} and SDEs
\citep{giles-or}, they can be applied also in 
other application domains, where estimates with gradually increasing
accuracy are available; see the recent review by 
\cite{giles-acta15} and references therein.

More recently, the so-called debiasing techniques
\citep{mcleish,rhee-glynn-wsc,rhee-glynn} have attracted a lot of
research activity
\citep[e.g.][]{glynn-rhee-mc,agapiou-roberts-vollmer,strathmann-sejdinovic-girolami,jacob-thiery,lyne-girolami-atchade-strathmann-simpson,walter},
although similar ideas have been suggested much earlier in more
specific contexts
\citep[e.g.][]{glynn-randomized,rychlik-unbiased,rychlik-kernel}. These
techniques are based on similar ideas as MLMC, but instead of optimal
allocation of computational resources for minimising the error, the
primary focus is on providing unbiased estimators. Monte Carlo
inference is straightforward with independent unbiased samples,
allowing to construct confidence intervals in a reliable way
\citep{glynn-whitt-sequential}. Debiasing techniques may also be employed
within a stochastic approximation algorithm \citep{dereich-mueller}. In
particular, in a stochastic gradient descent type algorithm
\citep{robbins-monro}
relevant for instance in maximum likelihood inference
\citep{delyon-lavielle-moulines},
unbiased gradient estimate implies pure
martingale noise, which is supported by a well-established
theory~\citep[e.g.][]{kushner-yin-sa,borkar-sa}. 

The debiasing techniques involve balancing with cost and variance,
which often boils down to similar methods and conditions as
those that are used with MLMC. 
The connection between MLMC and debiasing techniques has been pointed
out earlier at least by \cite{rhee-glynn}, 
\cite{giles-acta15} and \cite{dereich-mueller},
but this connection has not been fully explored yet. The purpose of
this paper is to further clarify the connection of MLMC and debiasing
techniques, within a general framework for unbiased estimators.

Many techniques for unbiased estimation other than those considered
here have been suggested in the literature. For instance, there is a
whole body of literature for `perfect sampling' by Markov chains
\citep{propp-wilson} 
or with certain classes of SDE models \citep{beskos-roberts}; 
see also the recent monograph by
\cite{huber-perfect} and references therein.
Perfect sampling can be used to construct unbiased estimators, but the
problem is generally more prestigious and often harder to implement.
See also the recent article by \cite{jacob-thiery} for discussion about
other related unbiased estimation techniques.

The rest of the paper is organised as follows. The multilevel Monte
Carlo and the previous debiasing methods are presented briefly in
Sections \ref{sec:mlmc} and \ref{sec:unbiased}, respectively. Section \ref{sec:general}
introduces a new general unbiased scheme with an explicit expression
for its variance (Theorem \ref{thm:general}).  The unbiased estimators
suggested by \cite{mcleish} and \cite{rhee-glynn-wsc,rhee-glynn} are
reformulated as specific instances of this scheme, as well as an
obvious `hybrid' scheme with MLMC and unbiased components (Example
\ref{ex:hybrid-mlmc}).

New unbiased estimators are suggested in Section \ref{sec:new}. Two of
the new schemes, termed stratified and residual sampling estimators,
have provably lower variance than simple averages of independent
unbiased estimators (Proposition \ref{prop:consistency}). Because
stratification is a well-known variance reduction technique, these
estimators may be well-known, but they do not seem to be 
recorded in the literature yet. 
The first main finding of this
paper is Theorem \ref{thm:residual-eff} which shows that the
variances of two of the new schemes are asymptotically equal to that of
MLMC under general conditions. This result suggests that unbiasedness 
can often be achieved with asymptotically negligible additional variance.

The expected cost of the methods is discussed in Section
\ref{sec:cost}. The new schemes appear even more appealing after seeing 
that the expected cost of MLMC and
the unbiased schemes are also asymptotically equivalent (Proposition
\ref{prop:mlmc-expected-cost}) and therefore the efficiency of an estimator
with new schemes can be made arbitrarily close to MLMC (Corollary
\ref{cor:batch-complexity}). The limiting variance formulation 
in Theorem \ref{thm:residual-eff} leads into an easily applicable
optimisation criteria for the sampling distribution related to the
new estimators.

Section \ref{sec:stopping} presents a generalisation of the
unbiased scheme, which accomodates further conditioning and dependent
randomisation schemes based on stopping times.
Numerical experiments in Section \ref{sec:numerical} show how the
efficiency bounds predicted by theory are attained in four examples,
three of which were also studied by \cite{rhee-glynn}.  The paper is
concluded by a discussion about the implications of the findings in
Section \ref{sec:conclusions}. Some practical guidelines and possible
topics of further research are discussed as well.

We denote $a\wedge b \defeq \min\{a,b\}$, $a\vee b\defeq
\max\{a,b\}$, $(x)_+\defeq 0 \vee x$ and $\charfun{\uarg}$ stands for the indicator function.
The natural numbers are defined strictly positive
$\N\defeq\{1,2,\ldots\}$, empty sums $\sum_{j=1}^0(\uarg)$ are
taken as zero, and we use the convention $0/0=0$.
%}}}

%%%%%%%%%%%%%%%%%%%%%%%%%%%%%%%%%%%%%%%%%%%%%%%%%%%%%%%%%%%%%%%%%%%%%%%%%%%%%%
\section{Multilevel Monte Carlo} 
\label{sec:mlmc} %{{{

The starting point of MLMC is the existence of estimators $(Y_i)_{i\ge
1}$, which are increasingly accurate approximations of a given
`target' random variable $Y$, whose expectation is of interest. That
is, $(Y_i)_{i\ge 1}$ are random variables with $\E Y_i \to \E Y$, and
the task is to provide a numerical approximation of $\E Y$. The simulation
cost of a single realisation of $Y_i$ increases in $i$, which calls
for optimisation of computational resources.

Such a scenario arises for instance in the context of SDE models, 
commonly applied in option pricing. In such an application,
we might have $Y\defeq f(X_T)$, where 
$f:\R\to \R$ is a given `payoff' function, and 
$X_T$ corresponds to the terminal value of the 
process $(X_t)_{t\in[0,T]}$ solving the SDE
\begin{equation}
          \ud X_t = \mu(X_t) \ud t 
          + \sigma(X_t) \ud B_t, \qquad t\in[0,T], \qquad X_0\equiv
          x_0,
          \label{eq:sde}
\end{equation}
where $(B_t)_{t\ge 0}$ stands for the standard Brownian motion, 
the functions $\mu,\sigma:\R\to\R$ are the drift and diffusion parameters 
of the model, and
$x_0\in\R$ is the initial value. Unless $\mu$ and $\sigma$ are of
certain specific form, the random variable $X_T$ cannot be simulated
exactly, but admits easily implementable approximations based on 
time-discretisation of \eqref{eq:sde}. The commonly applied dyadic uniform 
meshes are defined as follows
\[
    \tau_i \defeq \big(0=t_0^{(i)}< t_1^{(i)} <\cdots < t_{2^i}^{(i)} = T\big)
          \qquad \text{where}
          \qquad t_j^{(i)} \defeq jT2^{-i}\qquad
          \text{for $i\ge 0$.}
\]
The \emph{Milstein} discretisation 
\citep[e.g.][]{kloeden-platen} corresponding to $\tau_i$ 
is defined by letting $X_{0}^{(\tau_i)}=x_0$ and iteratively calculating
\begin{align*}
    X_{k}^{(\tau_i)} &\defeq X_{k-1}^{(\tau_{i})}
    + \mu(X^{(\tau_i)}_{k-1}\big) \Delta t_k^{(i)}
    + \sigma\big(X^{(\tau_i)}_{k-1}\big)
    \Delta B_k^{(i)}
    + \frac{1}{2}\sigma\big(X^{(\tau_i)}_{k-1}\big)
    \sigma'\big(X^{(\tau_i)}_{k-1}\big)\big[ (\Delta B_k^{(i)})^2 -
    \Delta t_k^{(i)} \big],
\end{align*}
for $k=1,\ldots,2^i$, 
where $\Delta t_k^{(i)} \defeq t_k^{(i)}-t_{k-1}^{(i)}$
and $\Delta B_k^{(i)} \defeq B_{t_{k}^{(i)}}-B_{t_{k-1}^{(i)}}$
are independent $N(0,\Delta t_k^{(i)})$ random variables.
The final value of such an iteration provides an approximation of
$X_T$, so we may set 
$Y_{i+1} \defeq f(X_{2^i}^{(\tau_i)})$ for $i\ge 0$.
The cost of simulating a single realisation of $Y_i$ is of order $2^i$, and 
under certain, fairly general, conditions on 
$\mu,\sigma$ and $f$, the mean square error
$\E( Y - Y_{i-1})^2 \le c 2^{-\lambda i}$ with $\lambda>1$
\citep[e.g.][]{kloeden-platen}.

The MLMC is an efficient way to use such approximations 
in order to approximate $\E Y$.
It is based on the following, seemingly trivial observation
\[
    \E Y_m = \sum_{i=1}^m \E (Y_i - Y_{i-1}), \qquad\text{with}\qquad
      Y_0\defeq 0,
\]
which suggests that one may construct an estimate of $\E Y_m$ 
using a separate Monte Carlo approximation of each term $\E(Y_i-Y_{i-1})$:
\[
    Z_\ml \defeq \sum_{i=1}^m \frac{1}{n_i} \sum_{j=1}^{n_i} 
    \Delta_i^{(j)},
\]
where $n_1,\ldots,n_m\in\N$ and 
$\Delta_i^{(j)}$ are independent random variables with
$\E \Delta_i^{(j)} = \E Y_i - \E Y_{i-1}$.
It is direct to check that
then $\E Z_\ml = \E Y_m$. 
There are two key reasons why MLMC is useful:
\begin{enumerate}[(i)]
    \item The random variables $\Delta_i^{(j)}$ are usually
      independent realisations of
      $\Delta_i = Y_i - Y_{i-1}$, where the approximations
      $Y_i$ and $Y_{i-1}$ are \emph{dependent}, such that
      $|\Delta_i|$ is typically small when $i$ is large.
      In the context of the SDE example discussed above, 
      such a coupling arises naturally when using a
      common Brownian path in the discretisations leading to 
      $Y_i$ and $Y_{i-1}$.
    \item It allows to \emph{optimise} the computational effort 
      devoted to each `level' to estimate $\E(Y_i-Y_{i-1})$. The 
      key benefit is that
      fewer samples $n_i$ are often necessary 
      with higher $i$, which leads to lower overall cost.
\end{enumerate}
Theorem 1 of \cite{giles-acta15} quoted below gives the complexity of 
MLMC under the common exponential framework. It assumes that the expected cost
(computational complexity) of each term $\Delta_i^{(j)}$ 
is $\kappa_i$, so that the expected
cost of $Z_\ml$ is $\kappa = \sum_{i=1}^m n_i \kappa_i$.
%%%%%%%%%%%%%%%%%%%%
\begin{theorem} 
\label{thm:giles} %{{{
Suppose $(\Delta_i^{(j)})$ are independent, and there exist positive
$c_1,c_2,c_3,\alpha,\beta,\gamma$ with $\alpha \ge (\beta \wedge
\gamma)/2$, such that for all $i\ge 1$,
\[
    |\E Y - \E Y_i| \le c_1 2^{-\alpha i},
    \qquad
    \var(\Delta_i^{(j)}) \le c_2 2^{-\beta i},
    \qquad\text{and}\qquad
    \kappa_i \le c_3 2^{\gamma i}.
\]
Then, there exists a positive constant $c_4$ such that for any
$\epsilon<1/e$ there are values $m$ and $n_1,\ldots, n_m\in\N$ 
such that the MLMC estimator satisfies the following mean square error
(MSE) bound
\[
    \E(Z_\ml-\E Y)^2 \le \epsilon^{2},
\]
and satisfies the expected cost bound
\[
    \kappa \le 
    \begin{cases}
    c_4 \epsilon^{-2}, & \beta>\gamma, \\
    c_4 \epsilon^{-2}(\log \epsilon)^2, & \beta = \gamma, \\
    c_4 \epsilon^{-2-(\gamma-\beta)/\alpha}, & \beta < \gamma.
    \end{cases}
\]
\end{theorem} %}}}

The SDE application discussed above often satisfies the conditions of
Theorem \ref{thm:giles} with $\beta>\gamma$ \citep{kloeden-platen}. In a
multidimensional SDE setting, the antithetic truncated Milstein
\citep{giles-szpruch} often lead to $\beta>\gamma$ as well. This
highlights why MLMC has become so popular---it is often possible to
attain a `canonical' rate $\epsilon^{-2}$, equivalent to a square
root error rate $\kappa^{-1/2}$, despite the bias, using simple
standard discretisation methods. Crude Monte Carlo, that is, taking a
fixed level $m$ and averaging independent realisations of $Y_m$, leads
to worse rates.  The same canonical square root error rate can be attained
with similar assumptions using the debiasing techniques which are
discussed next.

%}}}

%%%%%%%%%%%%%%%%%%%%%%%%%%%%%%%%%%%%%%%%%%%%%%%%%%%%%%%%%%%%%%%%%%%%%%%%%%%%%%
\section{Debiasing techniques}
\label{sec:unbiased} %{{{

As with MLMC, assume that $(Y_i)_{i\ge 1}$ and $Y$ are integrable random
variables with $\E Y_i\to \E Y$ and $Y_0\equiv 0$, then
\[
    \E Y = \lim_{n\to\infty} \E Y_n = \lim_{n\to\infty} \sum_{i=1}^n
    \E(Y_i-Y_{i-1}).
\]
The fundamental observation behind the unbiased schemes is that one may
employ randomisation to pick a finite number of terms of the series
to construct estimators with expectation $\E Y$.

The two results below due to \cite{rhee-glynn} propose
two such estimators, along with expressions for their second moments.
Proofs of Theorems \ref{thm:single-term} and
\ref{thm:sum} are shown to follow as a consequence of Theorem \ref{thm:general}
in Section \ref{sec:general} (Examples \ref{ex:single} and \ref{ex:sum}).
%%%%%%%%%%%%%%%%%%%%
\begin{condition} 
    \label{cond:single-term} %{{{
Suppose $(\Delta_i)_{i\ge 1}$ are independent random variables with $\E\Delta_i = \E
Y_i - \E Y_{i-1}$, and $(p_i)_{i\ge 1}$ is a probability distribution
such that $p_i>0$ for all $i\in\N$ and
\[
    \sum_{i\ge 1}
    \frac{\E \Delta_i^2}{p_i}<\infty.
\]
\end{condition} %}}}

%%%%%%%%%%%%%%%%%%%%
\begin{theorem}[Single term estimator]
    \label{thm:single-term} %{{{
Suppose $(\Delta_i)_{i\ge 1}$ and $(p_i)_{i\ge 1}$ satisfy Condition
\ref{cond:single-term}, and $R\in\N$ is a random variable independent
of $(\Delta_i)_{i\ge 1}$ with
$\P(R=i)=p_i$,
then the \emph{single-term estimator} 
\[
    Z^{(1)} \defeq \frac{\Delta_R}{p_R}\qquad\text{satisfies}\qquad
    \E Z^{(1)}  = \E Y \quad\text{and}\quad
    \E (Z^{(1)})^2 = \sum_{i\ge 1}
    \frac{\E \Delta_i^2}{p_i}. 
\]
\end{theorem} %}}}

\cite{mcleish} first suggested the following estimators, but with 
different conditions and different expression for the variance.

%%%%%%%%%%%%%%%%%%%%
\begin{condition} 
    \label{cond:sum} %{{{
Suppose $(\Delta_i)_{i\ge 1}$ are 
random variables with $\E\Delta_i = \E
Y_i - \E Y_{i-1}$, and $(p_i)_{i\ge 1}$ is a probability distribution
such that $\tilde{p}_i \defeq \sum_{j\ge i} p_j>0$ for all $i\in\N$
and either one of the following hold:
\begin{enumerate}[(i)]
\item 
  \label{item:coupled-sum}
  $\Delta_i = Y_i - Y_{i-1}$ and
  \[
    \sum_{i\ge 1} \frac{\E(Y_{i-1} - Y)^2}{\tilde{p}_i} < \infty.
  \]
\item
  \label{item:independent-sum}
  $(\Delta_i)_{i\ge 1}$ are independent and 
\begin{equation*}
    \sum_{i\ge 1} \frac{\var(\Delta_i) + (\E Y - \E
Y_{i-1})^2}{\tilde{p}_i}
<\infty.
\end{equation*}  
\end{enumerate}
\end{condition} %}}}

%%%%%%%%%%%%%%%%%%%%
\begin{theorem}[Sum estimator] 
    \label{thm:sum} %{{{
Suppose $(\Delta_i)_{i\ge 1}$ and $(p_i)_{i\ge 1}$ satisfy Condition
\ref{cond:sum} and $R\in \N$ is a random variable independent of
$(\Delta_i)_{i\ge 1}$ with $\P(R=i)=p_i$, then
the \emph{sum estimator}
\[
    Z_\sumsymb^{(1)} \defeq \sum_{i=1}^R
    \frac{\Delta_i}{\tilde{p}_i}
\]
satisfies $\E Z_\sumsymb^{(1)} = \E Y$. In case of Condition
\ref{cond:sum} \eqref{item:coupled-sum},
\[
    \E (Z_\sumsymb^{(1)}) ^2 = \sum_{i\ge 1} 
    \frac{\E (Y_{i-1} - Y)^2  - \E( Y_i -
      Y)^2}{\tilde{p}_i},
\]
and, in case of Condition
\ref{cond:sum}
\eqref{item:independent-sum},
\[
    \E (Z_\sumsymb^{(1)}) ^2 = \sum_{i\ge 1}
      \frac{\var(\Delta_i) + (\E Y - \E Y_{i-1})^2 -
      (\E Y - \E Y_i)^2}{\tilde{p}_i}.
\]
\end{theorem} %}}}

%%%%%%%%%%%%%%%%%%%%
\begin{remark} %{{{
We follow \cite{rhee-glynn} and call
the estimator of Theorem \ref{thm:sum} 
\emph{coupled sum} in case of
Condition \ref{cond:sum} \eqref{item:coupled-sum} 
and \emph{independent sum} in case of 
Condition \ref{cond:sum} \eqref{item:independent-sum}.
Note that Condition \ref{cond:sum} \eqref{item:independent-sum}
is equivalent to the assumption of \cite[Theorem 5(a)]{rhee-glynn},
because $\tilde{p}_i$ is non-increasing, and so
\[
  \sum_{i\ge 1} \frac{(\E Y_i - \E Y_{i-1})^2}{\tilde{p}_i} 
  \le 2 \sum_{i \ge 1} \frac{(\E Y_i - \E Y)^2 + (\E Y_{i-1} - \E Y)^2}{\tilde{p}_i}
  \le 4 \sum_{i \ge 1} \frac{(\E Y_{i-1} - \E Y)^2}{\tilde{p}_i}.
\]
\end{remark} %}}}

Let us briefly return to the case where $Y_i=f(X_{2^i}^{(\tau_i)})$ and
$X_{2^i}^{(\tau_i)}$ corresponds to the solution of a time-discretised
SDE with a mesh of $2^i$ points as discussed in Section \ref{sec:mlmc}.
Many approximation schemes admit a weak rate $\alpha>1/2$,
which implies that $|\E Y_i - \E Y_{i-1}| \le |\E Y - \E Y_{i-1}|
+ |\E Y - \E Y_i|
\le c 2^{-\alpha i}$ \citep[cf.][]{kloeden-platen}.  
If additionally $\var(\Delta_i) \le c
2^{-\beta i}$ with $\beta>1$, often satisfied for instance by the
Milstein scheme \citep{kloeden-platen,jentzen-kloeden-neuenkirch} or in
the multivariate case by the antithetic truncated Milstein scheme
\citep{giles-szpruch}, then taking $p_i \propto 2^{-\xi i}$, where
$\xi\in\big(1,2\alpha\wedge\beta\big)$ satisfies Conditions
\ref{cond:single-term} and \ref{cond:sum}
\eqref{item:independent-sum}. Because the cost of level $i$ is 
of order $2^i$, this sampling scheme admits a 
finite expected cost (see Section \ref{sec:cost} for details).
This shows that unbiased estimators with finite variance and finite
expected cost can be obtained with the same conditions
under which MLMC admits the canonical error rate.
In case of correlated $(\Delta_i)_{i\ge 1}$ in 
Condition \ref{cond:sum} \eqref{item:coupled-sum}, 
a slightly more stringent condition about the strong rate
$\E( Y - Y_{i-1})^2 \le c 2^{-\lambda i}$ with $\lambda>1$ is required.

%}}}

%%%%%%%%%%%%%%%%%%%%%%%%%%%%%%%%%%%%%%%%%%%%%%%%%%%%%%%%%%%%%%%%%%%%%%%%%%%%%%
\section{General unbiased scheme} 
\label{sec:general} %{{{

\cite{giles-acta15} pointed out that averaging $n$ independent
single term estimators $Z^{(1)}$ introduced in Theorem
\ref{thm:single-term} corresponds to an estimator of the form
\[
    \sum_{i=1}^{\infty} \frac{1}{n p_i} \sum_{j=1}^{N_i}
    \big(Y_i^{(j)}-Y_{i-1}^{(j)}),
\]
where $N_i$ is the number of samples from level $i$, and 
the expectation of $N_i$ is $np_i$. This shows a close connection with
MLMC, which corresponds to taking $N_i\equiv n_i = np_i$ (up to some
level $m$). Inspired by this remark, and the techniques by 
\cite{rhee-glynn}, consider the following general unbiased
scheme, with an expression for its variance. 
%%%%%%%%%%%%%%%%%%%%
\begin{theorem} 
    \label{thm:general} %{{{
Suppose $(Y_i)_{i\ge 1}$ and $Y$ are integrable random variables
such that $\E Y_i \to \E Y$, and let $(\Delta_i)_{i\ge 1}$ be square
integrable random variables such that $\E\Delta_i =
\E Y_i - \E Y_{i-1}$ for all $i\ge 1$, with $\E Y_0 \defeq  0$.
Assume $(\Delta_i^{(j)})_{i\ge 1}$ for $j\ge 1$ are independent realisations of
the process $(\Delta_i)_{i\ge 1}$.
Suppose $N_i$ are non-negative square integrable random integers
independent of $(\Delta_i^{(j)})_{i,j}$ and with $\E N_i>0$ for all $i\ge 1$.
Define for $0\le \ell <m$
\begin{equation*}
    v_{\ell,m} \defeq 
    \sum_{i,k=\ell+1}^m 
    \frac{\cov(\Delta_i,\Delta_k)\E(N_i\wedge N_k) +
      \E\Delta_i\E\Delta_k\cov(N_i, N_k)}{\E N_i \E N_k} .
\end{equation*}
If there exists a strictly increasing sequence of positive integers 
$(m_k)_{k\ge 1}$ 
such that 
$\lim_{k\to\infty}\sup_{j\ge 1} v_{m_k,m_{k+j}}=0$ 
and $\sum_{i\ge 1} N_i<\infty$ almost surely,
then the estimator
\[
    Z \defeq \sum_{i=1}^\infty \frac{1}{\E N_i} \sum_{j=1}^{N_i}
    \Delta_i^{(j)} 
\]
is unbiased $\E Z = \E Y$ and $\var(Z) = \lim_{k\to\infty}
v_{0,m_k}<\infty$.
\end{theorem} %}}}
\noindent The proof of Theorem \ref{thm:general} is given in Appendix
\ref{app:proof-general}.

Consider first a simple example of Theorem \ref{thm:general},
where $N_i$ are taken independent.
%%%%%%%%%%%%%%%%%%%%
\begin{example}[Independent levels] %{{{
Let $(N_i)_{i\ge 1}$ be independent with $\mu_i \defeq \E N_i$ and
$\sum_{i\ge 1}\mu_i<\infty$, 
and suppose $\beta_i\defeq \var(N_i)/\mu_i <\infty$, and let 
$(\Delta_i)_{i\ge 1}$ be uncorrelated. Then,
\begin{align*}
    v_{\ell,m} &= \sum_{i=\ell+1}^m \frac{\var(\Delta_i) \E N_i +
      \var(N_i)(\E \Delta_i)^2}{(\E N_i)^2} 
    = \sum_{i=\ell+1}^m \frac{\var(\Delta_i) +
      \beta_i(\E\Delta_i)^2}{\mu_i }.
\end{align*}
If $\sum_i N_i<\infty$ a.s.~and $v\defeq \sum_{i\ge 1} [\var(\Delta_i) + \beta_i (\E
\Delta_i)^2]/\mu_i<\infty$, then
the corresponding estimator is unbiased with variance $v$.
In particular, taking $N_i$ Poisson with intensity $np_i$ where
$\sum_{i\ge 1}p_i=1$, then $\sum_i N_i <\infty$ by the Borel-Cantelli
lemma and $v = n^{-1} \sum_{i\ge 1} \E\Delta_i^2/p_i$.
\end{example} %}}}

The averages of single term estimators and sum estimators introduced
in Theorems \ref{thm:single-term} and \ref{thm:sum} correspond to
certain dependence structures of $(N_i)_{i\ge 1}$, as we shall see
next. The following two examples are important, because the new
schemes introduced in Section \ref{sec:new} will be based on the same
constructions.

%%%%%%%%%%%%%%%%%%%%
\begin{example}[Single term estimators] 
    \label{ex:single} %{{{
Suppose Condition \ref{cond:single-term} holds. Define
\[
    N_i \defeq \sum_{j=1}^n \charfun{R^{(j)} = i},
    \qquad
    \text{where $R^{(j)}$ are independent with $P(R^{(j)}=i) =
p_i$},
\]
and call the estimator $Z^{(n)}_\iid$. It is easy to see that this 
corresponds to an average of $n$ 
independent single term estimators $Z^{(1)}$.
Note that $(N_i)_{i\ge 1}$ follow a multinomial distribution
with parameters $n$ and $(p_i)_{i\ge 1}$. We have $\E N_i = n p_i$ and
\[
    \E(N_i N_k)
    = \E \bigg[\sum_{j,\ell=1}^n
    \charfun{R^{(j)}=i}\charfun{R^{(\ell)}=k} \bigg]
    = n p_i \charfun{i=k} + n(n-1) p_i p_k,
\]
so $\cov(N_i, N_k) = n (p_i\charfun{i=k} - p_i p_k)$.
By the independence of $(\Delta_i)_{i\ge 1}$, 
\begin{align}
    v_{\ell,m}  
    &= \sum_{i=\ell+1}^m \frac{\var(\Delta_i)}{\E N_i}
    + \sum_{i,k=\ell+1}^m \frac{\E\Delta_i
      \E\Delta_k \cov(N_i, N_k)}{\E N_i \E N_k} \nonumber\\
    &= \frac{1}{n}\bigg(\sum_{i=\ell+1}^m \frac{\E\Delta_i^2}{p_i}
    - (\E Y_m - \E Y_{\ell})^2\bigg).
    \label{eq:v-single}
\end{align}
By assumption
$\sum_{i=\ell}^\infty \E\Delta_i^2/p_i<\infty$ and
$\E Y_i\to \E Y$, so
$\lim_{\ell\to\infty}\sup_{m>\ell} v_{\ell,m}=0$.
The variance satisfies
\[
    \var(Z^{(n)}_\iid) = \lim_{m\to\infty} v_{0,m}
    = \frac{1}{n}\bigg(\sum_{i=1}^\infty \frac{\E\Delta_i^2}{p_i}
    - (\E Y)^2\bigg),
\]
so $\var(Z^{(n)}_\iid) 
= n^{-1}\var(Z^{(1)})$. The above also proves
Theorem \ref{thm:single-term} with $n=1$.
\end{example} %}}}

%%%%%%%%%%%%%%%%%%%%
\begin{example}[Sum estimators] 
    \label{ex:sum} %{{{
Suppose $(\Delta_i^{(j)})_{i\ge 1}$ are independent realisations of
$(\Delta_i)_{i\ge 1}$ satisfying Condition \ref{cond:sum}, and
define $\tilde{N}_i \defeq \sum_{j=1}^n \charfun{R^{(j)}\ge i}$, 
where $R^{(j)}$ are independent with $\P(R^{(j)}= i)=p_i$. 
This estimator, which we
denote by $Z^{(n)}_{\sumsymb,\iid}$, 
corresponds to an average of $n$ independent sum estimators
$Z_\sumsymb^{(1)}$:
\[
    Z^{(n)}_{\sumsymb,\iid}
    = \sum_{i=1}^\infty \frac{1}{\E \tilde{N}_i} \sum_{j=1}^{\tilde{N}_i}
    \Delta_i^{(j)}
    \overset{\mathrm{d}}{=} \frac{1}{n} \sum_{j=1}^n
    \sum_{i=1}^{R^{(j)}} \frac{\Delta_i^{(j)}}{\tilde{p}_i},
\]
because $\E \tilde{N}_i = n \tilde{p}_i$, and where the latter
equality would hold if $(R^{(j)})$ were non-increasing. Changing the indexing
to ensure that $(R^{(j)})$ are non-increasing does not affect the
distribution. For all $i,k\ge 1$, 
\begin{align*}
    \E( \tilde{N}_i \tilde{N}_k )
    &= \sum_{j=1}^n \tilde{p}_{i\vee k} 
    + \sum_{\stackrel{j,\ell=1}{j\neq \ell}}^n
    \tilde{p}_i \tilde{p}_k
    = \E \tilde{N}_{i\vee k} + \frac{n-1}{n} \E \tilde{N}_i \E \tilde{N}_k,
\end{align*}
so $\cov(\tilde{N}_i, \tilde{N}_k) = \E \tilde{N}_{i\vee k} - n^{-1} \E \tilde{N}_i \E \tilde{N}_k$.
We also have $\E (\tilde{N}_i \wedge \tilde{N}_k) = \E \tilde{N}_{i\vee k}$
because $\tilde{N}_i\ge \tilde{N}_k$ for $i\le k$, so for $1\le \ell<m$,
\begin{align*}
    v_{\ell,m}
    &= \sum_{i,k=\ell+1}^m \frac{\E \tilde{N}_{i\vee k}\big(\cov(\Delta_i,\Delta_k)
      + \E \Delta_i \E \Delta_k\big)}{\E \tilde{N}_i \E \tilde{N}_k}
      - \frac{1}{n}\sum_{i,k=\ell+1}^m \E \Delta_i \E \Delta_k \\
    &= \sum_{i=\ell+1}^m \bigg( \frac{\E \Delta_i^2}{\E \tilde{N}_i} 
    + 2 \sum_{k=i+1}^m 
    \frac{\E( \Delta_i \Delta_k)}{\E \tilde{N}_i} 
    \bigg) -
     \frac{1}{n}(\E Y_m - \E Y_{\ell})^2.
\end{align*}

Denote for the rest of the proof $D_{i,m} \defeq Y_i - Y_m$.
If Condition \ref{cond:sum} \eqref{item:coupled-sum} holds,
we obtain
\begin{align*}
v_{\ell,m}
%    &= \frac{1}{n} \sum_{i=\ell+1}^m \frac{\E(Y_{i-1}-Y_m)^2 -
%      \E(Y_i-Y_m)^2}{\tilde{p}_i}
%    -
%     \frac{1}{n}(\E Y_m - \E Y_{\ell})^2 ,
    &= \frac{1}{n} \sum_{i=\ell+1}^m \frac{\E D_{i-1,m}^2 -
      \E D_{i,m}^2}{\tilde{p}_i}
    -
     \frac{1}{n}(\E D_{m,\ell})^2 ,
\end{align*}
%because $\Delta_i^2 + 2\Delta_i (Y_m - Y_i) = (Y_{i-1}-Y_m)^2 -  (Y_i-Y_m)^2$.
because $\Delta_i^2 + 2\Delta_i D_{m,i} = D_{i-1,m}^2 -  D_{i,m}^2$.
Let $(m_k)_{k\ge 1}$ be from Lemma \ref{lem:subsequence} in
Appendix \ref{app:single-sum}, then for any $k,j\ge 1$
\[
    v_{m_k,m_{k+j}} 
    \le \frac{4}{n} \sum_{i=m_{k}+1}^\infty \frac{\E(Y_{i-1}-Y)^2}{\tilde{p}_i}
%    - \frac{1}{n}(\E Y_{m_{k+j}}-\E Y_{m_k})^2,
    - \frac{1}{n}(\E D_{m_{k+j},m_k})^2,
\]
which converges to zero as $k\to\infty$. For any fixed $k\ge 1$, 
we have $v_{0,m_{k+j}}\to \var(Z_{\sumsymb,\iid}^{(n)})$ as
$j\to\infty$, so
\begin{align*}
    \var(Z_{\sumsymb,\iid}^{(n)}) 
%    &= \lim_{k\to\infty} v_{0,m_k}  \\
    &= \frac{1}{n} \lim_{j\to\infty} \bigg(
    \sum_{i=1}^{m_k} \frac{\E D_{i-1,m_{k+j}}^2 -
      \E D_{i,m_{k+j}}^2}{\tilde{p}_i} + (\E Y_{m_{k+j}})^2
      + \sum_{i=m_k+1}^{m_{k+j}}
      \frac{\E D_{i-1,m_{k+j}}^2  -
      \E D_{i,m_{k+j}}^2}{\tilde{p}_i}
      \bigg).
\end{align*}
The latter sum is upper bounded by $4\sum_{i\ge m_k+1}
\E(Y_{i-1}-Y)^2/\tilde{p}_i$, and because $\E D_{i,m_{k+j}}^2\to \E
(Y_i - Y)^2$ and $\E Y_{m_{k+j}}\to \E Y$ as $j\to\infty$, 
we may conclude that
\begin{align*}
    \var(Z_{\sumsymb,\iid}^{(n)}) 
    &=\frac{1}{n}\bigg(
    \sum_{i=1}^{m_k} 
    \frac{\E (Y_{i-1}-Y)^2 - \E (Y_i - Y)^2}{\tilde{p}_i}
    -(\E Y)^2
    \bigg)
    = \frac{1}{n} \var(Z_\sumsymb^{(1)}).
\end{align*}

Suppose then that Condition \ref{cond:sum} \eqref{item:independent-sum}
holds.
It is straightforward to check that 
\begin{align*}
    v_{\ell,m}
    = \frac{1}{n} \sum_{i=\ell+1}^m \frac{\var(\Delta_i) + (\E
      Y_{i-1}-\E Y_m)^2 -
      (\E Y_i-\E Y_m)^2}{\tilde{p}_i}
    -
     \frac{1}{n}(\E Y_m - \E Y_{\ell})^2 ,
\end{align*}
which satisfies $\lim_{\ell\to\infty}\sup_{m>\ell} v_{\ell,m}=0$ by
assumption, and similarly as above, 
\[
    \var(Z_{\sumsymb,\iid}^{(n)}) 
    = \lim_{k\to\infty} v_{0,k} = \frac{1}{n} \bigg( \sum_{i=1}^\infty
    \frac{\var(\Delta_i) + (\E
      Y_{i-1}-\E Y)^2 -
      (\E Y_i-\E Y)^2}{\tilde{p}_i} - (\E Y)^2 \bigg).
\]
Taking $n=1$ concludes also the proof of Theorem \ref{thm:sum}.
\end{example} %}}}

The last example is an obvious `hybrid' scheme involving
MLMC and an `unbiased tail' scheme, which also falls into the framework
of Theorem \ref{thm:general}.
%%%%%%%%%%%%%%%%%%%%
\begin{example} 
    \label{ex:hybrid-mlmc} %{{{
Assume that $r,m\in\N$ and that 
$N_i\equiv n_i\in\N$ for $i=1,\ldots,m$, and
$(p_i)_{i>m}$ are positive with $\sum_{i>m} p_i=1$, let 
$R^{(j)}$ be independent with $\P(R^{(j)}=i)=p_i$, and define
$N_i \defeq \sum_{j=1}^{r} \charfun{R^{(j)}=i}$
for $i>m$. The estimator can be written as
\[
    Z = \sum_{i=1}^m \frac{1}{n_i}\sum_{j=1}^{n_i}
    \Delta_i^{(j)} + \sum_{i=m+1}^\infty \frac{1}{r p_i}\sum_{j=1}^{N_i}
    \Delta_i^{(j)}.
\]
The first term coincides with an MLMC estimator with $m$
levels, with expectation $\E Y_m$, and
the second term is an average of single term estimators of $\E Y - \E Y_m$, with
$r$ samples. Note that we could have used any unbiased scheme in the latter
part, provided it satisfies the conditions in Theorem \ref{thm:general}.
\end{example} %}}}

The new sampling schemes discussed next in Section \ref{sec:new}
provide a different view on the balancing; see in particular Theorem
\ref{thm:residual-eff}.

%}}}

%%%%%%%%%%%%%%%%%%%%%%%%%%%%%%%%%%%%%%%%%%%%%%%%%%%%%%%%%%%%%%%%%%%%%%%%%%%%%%
\section{New unbiased estimators} 
\label{sec:new} %{{{

Let us now turn into new practically interesting estimators which
correspond to specific choices in Theorem \ref{thm:general}. The
estimators are based on stratification, a classical variance reduction
technique in survey sampling \citep[e.g.][]{hansen-hurwitz-madow}, 
which have also been widely used in Monte Carlo; see for instance
\cite{glasserman,douc-cappe-moulines}.

The following lemma states classical results on stratification
\cite[cf.][for proof]{vihola-unbiased}.
%%%%%%%%%%%%%%%%%%%%
\begin{lemma} 
    \label{lem:strat} %{{{
Let $X$ be an integrable random variable and let $A_1,\ldots,A_m$ be
exhaustive disjoint
events with 
$\P(A_i)=q_i>0$, and let
$\ell_1,\ldots,\ell_m\in\N$. Assume
$X_i^{(j)}$ are random variables with conditional
laws $\P(X_i^{(j)}\in B) = \P(\{X\in B\}\cap A_i)/q_i$.
The stratified estimator 
  $\tilde{X}\defeq \sum_{i=1}^m (q_i/\ell_i)\sum_{j=1}^{\ell_i}
    X_i^{(j)}$  satisfies
\begin{enumerate}[(i)]
\item \label{item:strat-pessimistic} Unbiasedness
    $\E\tilde{X} = \E X$ and $\var(\tilde{X})\le \var(X)$.
\item \label{item:prop-alloc}
  If $\{X_i^{(j)}\}_{i,j}$ are independent
and $\ell_i = q_i \ell$ (proportional allocation), then $\tilde{X} = \ell^{-1} 
\sum_{i=1}^m \sum_{j=1}^{\ell_i} X_i^{(j)}$ and
$\var(\tilde{X})  \le \ell^{-1}\var(X)$. 
\end{enumerate}
\end{lemma} %}}}

In what follows,
assume $(p_i)_{i\ge 1}$ are positive and such that $\sum_{i\ge 1}
p_i=1$, and denote $\tilde{p_i} \defeq \sum_{j\ge i} p_j$.
Let us first consider uniform stratification schemes.
For that purpose, recall the definition of the generalised inverse distribution function 
$F^{-1}: (0,1)\to\N$ corresponding to $(p_i)_{i\ge 1}$,
\[
    F^{-1}(u) \defeq \min\{k\in\N\given F(k)\ge u\}
    \qquad\text{with}\qquad \textstyle F(k) \defeq \sum_{i=1}^k p_i.
\]
The well-known inverse distribution function method
states that a uniform $U\sim U(0,1)$ is transformed by $F^{-1}$ into $R=F^{-1}(U)$
with $\P(R=i)=p_i$.

%%%%%%%%%%%%%%%%%%%%
\begin{definition}[Uniformly stratified estimators] 
    \label{def:stratified} %{{{
Let $n\in\N$, and assume $U^{(j)}$ are independent 
$U\big(\frac{j-1}{n},\frac{j}{n}\big)$ random variables for $j=1,\ldots,n$.
Let $R^{(j)} = F^{-1}(U^{(j)})$.
\begin{enumerate}[(i)]
\item \label{item:strat-single}
Define $N_i \defeq \sum_{j=1}^n \charfun{R^{(j)}=i}$,
then the estimator
$Z_\str^{(n)}$ defined as in Theorem 
\ref{thm:general} is the \emph{uniformly stratified single term
  estimator}.
\item \label{item:strat-sum}
  Define $\tilde{N}_i \defeq \sum_{j=1}^n \charfun{R^{(j)}\ge i}$,
  then the estimator
  $Z_{\sumsymb,\str}^{(n)}$ defined as in 
  Theorem \ref{thm:general} is the \emph{uniformly stratified sum
    estimator}.
\end{enumerate}
\end{definition} %}}}

%%%%%%%%%%%%%%%%%%%%
\begin{definition}[Systematic sampling estimators] 
\label{def:systematic} %{{{
Let $u_j \defeq (j-1)/n$ and define $U^{(j)} \defeq
u_j + U$ for all $j=1,\ldots,n$, where $U\sim U(0,1/n)$, 
and let $R^{(j)} \defeq F^{-1}(U^{(j)})$.
Define then $N_i$ and $\tilde{N}_i$ as in Definition
\ref{def:stratified} \eqref{item:strat-single} and
\eqref{item:strat-sum}, respectively, and
the corresponding \emph{systematic sampling single-term estimator}
$Z_\sys^{(n)}$ and \emph{systematic sampling sum estimator}
$Z_{\sumsymb,\sys}^{(n)}$.
\end{definition} %}}}

The consistency and a variance bound for the uniformly
stratified and systematic sampling estimators are stated in Proposition
\ref{prop:consistency}, after introducing another slightly different 
stratification scheme. 

%%%%%%%%%%%%%%%%%%%%
\begin{definition}[Residual sampling estimators] 
\label{def:residual} %{{{
Let $n\in\N$, define $n_i \defeq \lfloor n p_i \rfloor$
and let $r \defeq n - \sum_{i\ge 1} n_i\ge 0$. If $r>0$, define
the `residual' probability distribution $(p_i^*)_{i\ge 1}$ as 
$p_i^* \defeq (n p_i - n_i)/r$,
and let $Q^{(j)}$ be independent random variables such that 
$\P(Q^{(j)}=i)=p_i^*$ for $j=1,\ldots,r$. 
\begin{enumerate}[(i)]
\item 
Define $N_i \defeq n_i + N_i^*$ where 
    $N_i^* \defeq \textstyle\sum_{j=1}^r \charfun{Q^{(j)}=i}$,
then the estimator
$Z_\res^{(n)}$ defined as in Theorem 
\ref{thm:general} is the \emph{residual sampling single term estimator}.
\item Define 
   $\tilde{N}_i \defeq \tilde{n}_i + \tilde{N}_i^*$, where
   $\tilde{n}_i \defeq \textstyle\sum_{k=1}^i n_i$ and
   $\tilde{N}_i^* \defeq \textstyle\sum_{j=1}^r \charfun{Q^{(j)}\ge
     i}$,
  then the estimator
  $Z_{\sumsymb,\res}^{(n)}$ defined as in 
  Theorem \ref{thm:general} is the \emph{residual sampling sum
    estimator}.
\end{enumerate}
\end{definition} %}}}

%%%%%%%%%%%%%%%%%%%%
\begin{proposition} 
\label{prop:consistency} %{{{
The estimators in Definitions \ref{def:stratified},
\ref{def:systematic} and \ref{def:residual} satisfy:
\begin{enumerate}[(i)]
\item Assume Condition \ref{cond:single-term}, then
\begin{align*}
    \E Z_\str^{(n)} 
    = \E Z_\res^{(n)} 
    &=\E Y, &
    \var(Z_\str^{(n)})\vee 
    \var(Z_\res^{(n)})
    &\le n^{-1} \var(Z^{(1)}), \\
    \E Z_\sys^{(n)}  & = \E Y,
    &
    \var(Z_\sys^{(n)})
    &\le \var(Z^{(1)}).
\end{align*}
\item Assume Condition \ref{cond:sum}, then
\begin{align*}
    \E Z_{\sumsymb,\str}^{(n)} 
    = \E Z_{\sumsymb,\res}^{(n)} 
    &=\E Y, &
    \var(Z_{\sumsymb,\str}^{(n)})\vee 
    \var(Z_{\sumsymb,\res}^{(n)})
    &\le n^{-1} \var(Z_\sumsymb^{(1)}), \\
    \E Z_{\sumsymb,\sys}^{(n)}  & = \E Y,
    &
    \var(Z_{\sumsymb,\sys}^{(n)})
    &\le \var(Z_\sumsymb^{(1)}).
\end{align*}
\end{enumerate}
\end{proposition} %}}}
\noindent Proposition \ref{prop:consistency} follows as a consequence of 
Lemma \ref{lem:strat}; the detailed proof is given in Appendix
\ref{app:proof-prop-consistency}.

%%%%%%%%%%%%%%%%%%%%
\begin{remark} %{{{
Instead of using $u_j = (j-1)/n$ in systematic sampling, we could use 
instead any sequence $u_1,\ldots,u_n\in[0,1]$, and let $U^{(j)} \defeq
u_j + U \mod 1$, where $U\sim U(0,1)$. This would not be
stratification, but it is direct to check the estimators still retain
the same expectation, and also the same pessimistic variance bound.
For instance, using a low-discrepancy sequence $(u_j)$ would
correspond to randomised quasi-Monte Carlo \citep[e.g.][]{dick-kuo-sloan}.
\end{remark} %}}}

Proposition \ref{prop:consistency} states that all the new estimators
are unbiased, and that the uniformly stratified and residual
sampling estimators cannot be worse than averages of independent
single term and sum estimators, in terms of variance. In fact, they often have a
strictly lower variance, but it is generally difficult to quantify the
benefit. The following Theorem \ref{thm:residual-eff} is often more useful,
indicating that stratified and residual sampling schemes attain asymptotically the
efficiency of MLMC under general conditions.

Let us first formulate `idealised' MLMC strategies 
based on limiting allocations
$(p_i)_{i\ge 1}$. 
%%%%%%%%%%%%%%%%%%%%€
\begin{definition} 
    \label{def:mlmc} %{{{
Assume $(p_i)_{i\ge 1}$ are non-negative 
with $\sum_{i\ge 1} p_i = 1$ and define $\tilde{p}_i = \sum_{j\ge i}
p_i$. For any $n\in\N$, define $n_i \defeq \lfloor n p_i \rfloor$, 
$\hat{n}_i \defeq \lfloor n \tilde{p}_i \rfloor$, 
$m_n = \max\{i\ge 0\given
  n_i>0\}$ and $\tilde{m}_n = \max\{i\ge 0\given
  \hat{n}_i>0\}$ and denote
\[
    Z_\ml^{(n)}\defeq 
    \sum_{i=1}^{m_n} \frac{\charfun{n_i>0}}{n_i}
    \sum_{j=1}^{n_i} \Delta_i^{(j)},
    \qquad\text{and}\qquad
    Z_{\sumsymb,\ml}^{(n)}\defeq 
    \sum_{i=1}^{\tilde{m}_n} \frac{1}{\hat{n}_i}
    \sum_{j=1}^{\hat{n}_i} \Delta_i^{(j)}.
\]
\end{definition}
%}}}

Practical implementations of MLMC are often based on application of
stopping rules, which may determine $n_i$ and $m_n$ during simulation.
Definition \ref{def:mlmc} can be therefore viewed as `idealised'
version of MLMC,  where (nearly) optimal allocation strategy
$(p_i)_{i\ge 1}$ is known beforehand, and
$n_i$ and $m_n$ are determined in terms of a single `running-time'
parameter $n$.

%%%%%%%%%%%%%%%%%%%%€
\begin{theorem} 
    \label{thm:residual-eff} %{{{
Below, $Z_*^{(n)}$ (resp.~$Z_{\Sigma,*}^{(n)}$) 
stands for either $Z_\res^{(n)}$ or $Z_\str^{(n)}$
(resp.~$Z_{\Sigma,\res}^{(n)}$ or $Z_{\Sigma,\str}^{(n)}$).
\begin{enumerate}[(i)]
    \item \label{item:residual-eff-single}
      Assume Condition \ref{cond:single-term}
      then
\[
    \lim_{n\to\infty} n \var(Z_\ml^{(n)})
    = \lim_{n\to\infty} n \var(Z_*^{(n)})
    = \sigma_\infty^2  \defeq \sum_{i=1}^\infty
\frac{\var(\Delta_i)}{p_i}.
\]
\item \label{item:residual-eff-sum}
  Assume Condition \ref{cond:sum} \eqref{item:coupled-sum}, then
\[
    \lim_{n\to\infty} n \var(Z_{\sumsymb,*}^{(n)})
    = \sigma_{\sumsymb,\infty}^2  \defeq \sum_{i=1}^\infty
\frac{\var(Y-Y_{i-1}) - \var(Y-Y_i)}{\tilde{p}_i},
\]
and if $\sum_{i\ge 1}
\sup_{k\ge i} \var(Y_{k-1}-Y)/\tilde{p}_k<\infty$, then
$\lim_{n\to\infty} n \var(Z_{\sumsymb,\ml}^{(n)}) =
\sigma_{\sumsymb,\infty}^2$.
\item \label{item:eff-independent-sum}
  Assume Condition \ref{cond:sum} \eqref{item:independent-sum}, then
\[
    \lim_{n\to\infty} n \var(Z_{\sumsymb,\ml}^{(n)})
    = \lim_{n\to\infty} n \var(Z_{\sumsymb,*}^{(n)})
    =
    \sigma_{\sumsymb,\infty}^2 
    =  \sum_{i\ge 1} \frac{\var(\Delta_i)}{\tilde{p}_i}.
\]
\end{enumerate}
\end{theorem} %}}}
\noindent The proof of Theorem \ref{thm:residual-eff} is given in Appendix
\ref{app:proof-eff}.

%%%%%%%%%%%%%%%%%%%%
\begin{remark} 
    \label{rem:eff-diff} %{{{
The limiting variances in Theorem \ref{thm:residual-eff} can be
significantly
lower than the upper bounds given in Proposition \ref{prop:consistency},
which correspond to averaging independent unbiased estimators.
Indeed, $\var(\Delta_i) = \E\Delta_i^2 -
(\E\Delta_i)^2$, so
\[
    \sigma_\infty^2 = \var(Z^{(1)}) - \bigg[ \sum_{i\ge 1}
    \frac{(\E\Delta_i)^2}{p_i} - (\E Y)^2\bigg].
\]
Likewise, in case of Condition \ref{cond:sum} \eqref{item:coupled-sum}, $\var(Y-Y_i) = \E(Y-Y_i)^2 - (\E Y - \E Y_i)^2$, so
\begin{align*}
    \sigma_{\sumsymb,\infty}^2
    &= \E\big(Z_\sumsymb^{(1)})^2 
    - \sum_{i\ge 1} \frac{(\E Y - \E Y_{i-1})^2 - (\E Y - \E
      Y_i)^2}{\tilde{p}_i} \\
    &= \var(Z_\sumsymb^{(1)}) - 
    \sum_{i\ge 1} (\E Y - \E Y_i)^2
    \bigg(\frac{1}{\tilde{p}_{i+1}}-\frac{1}{\tilde{p}_i}\bigg).
\end{align*}
Note that $\tilde{p}_{i+1}\le \tilde{p}_i$ for all $i$, so all
the terms in the sum are positive. 
\end{remark} %}}}

%}}}

%%%%%%%%%%%%%%%%%%%%%%%%%%%%%%%%%%%%%%%%%%%%%%%%%%%%%%%%%%%%%%%%%%%%%%%%%%%%%%
\section{Expected cost and asymptotically optimal distribution}
\label{sec:cost} %{{{

Let us next consider the cost of simulating estimators $Z$ of the form
given in Theorem \ref{thm:general}. Assume each 
$\Delta_i^{(j)}$ has random cost $K_{i,j}$, 
such that the total cost of $Z$ is $K \defeq \sum_{i\ge 1} \sum_{j=1}^{N_i}
K_{i,j}$. Assume also that $\{K_{i,j}\}_{i,j\ge 1}$ are independent of 
$(N_i)_{i\ge 1}$ and
$\{K_{i,j}\}_{j\ge 1}$ have a common mean $\kappa_i = \E K_{i,j} <\infty$. 
The expected cost can be written down as follows.
\[
    \E K = \sum_{i\ge 1} \mu_i \kappa_i
    \qquad\text{with}\qquad
    \mu_i \defeq \E N_i.
\]
The following records the immediate fact 
that the estimators 
introduced in Section \ref{sec:new} have
the same expected cost as the simple averages of independent estimators.
%%%%%%%%%%%%%%%%%%%%
\begin{proposition} 
    \label{prop:cost-new} %{{{
Denote the cost of the estimator
$Z_*^{(n)}$ by $K_*^{(n)}$, where  `$*$' is a place holder.
Then,
\begin{enumerate}[(i)]
\item $\E K_\str^{(n)}
    = \E K_\res^{(n)}
    = \E K_\sys^{(n)}
    = \E K_\iid^{(n)} = n \E K^{(1)} = n \sum_{i\ge 1} p_i \kappa_i$,
\item $\E K_{\sumsymb,\str}^{(n)}
    = \E K_{\sumsymb,\res}^{(n)}
    = \E K_{\sumsymb,\sys}^{(n)}
    = \E K_{\sumsymb,\iid}^{(n)} = n \E K_\sumsymb^{(1)}
    = n \sum_{i\ge 1} \tilde{p}_i \kappa_i$.
\end{enumerate}
\end{proposition} %}}}
The following proposition records that both of the MLMC estimators 
introduced in Definition \ref{def:mlmc} 
have a cost that is asymptotically equivalent to 
the corresponding unbiased estimators, which are based on 
the same limiting allocation strategy.
%%%%%%%%%%%%%%%%%%%%
\begin{proposition} 
    \label{prop:mlmc-expected-cost} %{{{
  The expected cost of MLMC estimators in Theorem \ref{thm:residual-eff}
  satisfy
  \[
      \lim_{n\to\infty} n^{-1} \E K_\ml^{(n)}
      = {\textstyle \sum_{i\ge 1}} p_i \kappa_i
      \qquad\text{and}\qquad
      \lim_{n\to\infty} n^{-1} \E K_{\sumsymb,\ml}^{(n)}
      = {\textstyle \sum_{i\ge 1}} \tilde{p}_i \kappa_i.
  \]
\end{proposition} %}}}
\proof{Proof.} %{{{
The expected cost of $Z_\ml^{(n)}$ is 
\[
    \textstyle
    \sum_{i\ge 1} \charfun{n_i>0} n_i \kappa_i, \qquad
    \text{where}\qquad 
    n_i \defeq \lfloor n p_i \rfloor.
\]
Clearly $n_i/n\le p_i$, $n_i/n\to p_i$ and $\charfun{n_i>0}\to 1$ as
$n\to\infty$, 
so if $\sum_{i\ge 1}^\infty p_i \kappa_i <\infty$, the result follows by 
dominated convergence; otherwise one can use Fatou's lemma.
The proof with $\tilde{K}_\ml^{(n)}$ is identical. \Halmos
\endproof %}}}

\cite{glynn-whitt} formulate an asymptotic efficiency
principle, which states that if $(Z^{(j)})_{j\ge 1}$ are
i.i.d.~estimators with finite variance
$\sigma^2<\infty$ and with expected cost $\kappa<\infty$, then the
average of such estimators has asymptotic relative 
efficiency $[\kappa \sigma^2]^{-1}$.
Considering this notion of efficiency,
let $n\in\N$, and consider the following estimator, which is 
the average of $m$ independent realisations of the estimators above,
\begin{equation}
    \mathcal{Z}_*^{(n,m)} \defeq 
    \frac{1}{m} \sum_{j=1}^m Z_*^{(n,j)},
    \label{eq:avg-iid-unbiased}
\end{equation}
where $\{Z_*^{(n,j)}\}_{j\ge 1}$ 
are independent realisations of $Z_*^{(n)}$
and `$*$' is a place holder.
By letting $m\to\infty$, one may consider the asymptotic efficiency 
of this estimator as in the following
result, stating that stratified and residual sampling procedures
always improve upon averages of independent single term and sum
estimators, and can be made arbitrarily close to MLMC in asymptotic efficiency.

%%%%%%%%%%%%%%%%%%%%
\begin{corollary} 
    \label{cor:batch-complexity} %{{{
Suppose $n\in\N$ is fixed, then the following asymptotic efficiency
holds when $m\to\infty$.
\begin{enumerate}[(i)]
    \item If Condition \ref{cond:single-term} holds and
      $\sum_{i\ge 1} p_i \kappa_i<\infty$, then 
      the asymptotic efficiency of both
      $\mathcal{Z}_{\str}^{(n,m)}$ and $\mathcal{Z}_{\res}^{(n,m)}$ 
      is no worse than 
            that of $Z^{(mn)}_{\iid}$, and
can be made arbitrarily close to that of $Z_\ml^{(mn)}$ 
     by choosing $n$ sufficiently large.
    \item If Condition \ref{cond:sum} holds and
      $\sum_{i\ge 1} \tilde{p}_i \kappa_i<\infty$,
         then the asymptotic efficiency of both
      $\mathcal{Z}_{\sumsymb,\str}^{(n,m)}$ and 
      $\mathcal{Z}_{\sumsymb,\res}^{(n,m)}$ 
      is no worse than 
            that of $Z_{\sumsymb,\iid}^{(mn)}$, and
can be made arbitrarily close to that of $Z_{\sumsymb,\ml}^{(mn)}$
     by choosing $n$ sufficiently large.
\end{enumerate}
\end{corollary} %}}}
%%%%%%%%%%%%%%%%%%%%
\begin{remark} %{{{
Strictly speaking, the use of asymptotic efficiency principle requires that
the multilevel estimators $Z_{\sumsymb,\iid}^{(mn)}$ would be guaranteed to 
satisfy a functional central limit theorem, which is out of the scope
of this work; see, however, the recent works
of \cite{alaya-kebaier} and \cite{zheng-glynn} in this
direction in a different setting.
\end{remark} %}}}

If we assume $n$ is taken sufficiently large so that 
the limits in Theorem \ref{thm:residual-eff} are approximately attained,
the asymptotic efficiency principle suggests 
a rule for tuning the distribution $(p_i)$ 
for stratified and residual sampling distributions: the distribution 
$(p_i)_{i\ge 1}$ which minimises 
the asymptotic inverse relative efficiency (IRE)
$\sigma_\infty^2 \E K_*^{(n)}$ or 
$\sigma_{\sumsymb,\infty}^2 \E K_{\sumsymb,*}^{(n)}$ maximises the
efficiency. For the single term estimator, Condition
\ref{cond:single-term}, this leads to
\begin{equation}
    \min_{(p_i)_{i\ge 1}}
    \bigg(\sum_{i=1}^\infty
    \frac{\var(\Delta_i)}{p_i}\bigg)\bigg(\sum_{i=1}^\infty p_i
    \kappa_i\bigg).
    \label{eq:efficiency-criterion}
\end{equation}
The solution to 
\eqref{eq:efficiency-criterion} is proportional to $\beta_i \defeq
\sqrt{\var(\Delta_i)/\kappa_i}$, if $\sum_i \beta_i<\infty$ 
\citep[Proposition 1]{rhee-glynn}. 
This is
straightforward to implement in practice, because reliable estimates
of $\var(\Delta_i)$ are easily available. A straightforward practical
procedure, also implemented in the experiments below, 
is to use the `empirical optimal 
distribution' $\beta_i/\sum_i\beta_i$ to define directly the first $m$
probabilities $p_1,\ldots,p_m$ and a tail probability $\sum_{j>m} p_j$.
The tail probabilities $p_j$ for $j>m$ are then defined to follow a 
parametric distribution which guarantees a finite variance
based on theory. In the experiments, the tail distribution
was chosen to be geometric.

In case of the independent sum estimator, Condition \ref{cond:sum}
\eqref{item:independent-sum}, similar minimisation
\eqref{eq:efficiency-criterion} with $\tilde{p}_i$ in place of $p_i$
yields the asymptotically optimal distribution $p_i =
\tilde{p_i}-\tilde{p}_{i+1}$.
However, contrary to the single-term estimator, $(\tilde{p}_i)_{i\ge
1}$ must additionally be non-increasing. Because
\eqref{eq:efficiency-criterion} is invariant under multiplicative
constants on $(p_i)_{i\ge 1}$, the independent sum estimator can never
outperform the single term estimator, in terms of asymptotic IRE.

The coupled sum estimator, Condition \ref{cond:sum}
\eqref{item:coupled-sum}, leads to the optimisation problem
\begin{equation*}
    \min_{(\tilde{p}_i)_{i\ge 1}}
    \bigg(\sum_{i=1}^\infty
    \frac{\var(Y-Y_{i-1})-\var(Y-Y_i)}{\tilde{p}_i}\bigg)\bigg(\sum_{i=1}^\infty
    \tilde{p}_i
    \kappa_i\bigg).
%    \label{eq:efficiency-coupled}
\end{equation*}
This is more involved for two reasons. Estimation of the terms
$\var(Y-Y_i)$ is not straightforward. In practice, a reasonable
(but expensive) approximation may be obtained by approximating $Y$ with $Y_m$, where
$m$ is large. As noted in \cite[§3]{rhee-glynn}, 
the coupled sum estimator offers one more degree of freedom.
Namely, it is possible to consider a coupled sum estimator for a
subsequence of the variables $(Y_i)$, that is,
employing level differences 
$\tilde{\Delta}_i^{(J)} \defeq
Y_{J_i} - Y_{J_{i-1}}$, where $J = (0=J_0 < J_1<\cdots)$.
Then, the optimisation would require solving
\begin{equation*}
    \min_{J} 
    \min_{(\tilde{p}_i)_{i\ge 1}}
    \bigg(\sum_{i=1}^\infty
    \frac{\var(Y-Y_{J_{i-1}})-\var(Y-Y_{J_i})}{\tilde{p}_i}\bigg)\bigg(\sum_{i=1}^\infty
    \tilde{p}_i
    \kappa_{J_i}\bigg).
%    \label{eq:efficiency-coupled}
\end{equation*}
This leads to a combinatorial problem, which is discussed in depth in
\cite[§3]{rhee-glynn}, who also describe a dynamic programming 
algorithm which finds such optimal $J$, up to an index $m$, in
$O(m^3)$ time.

The asymptotic optimality results above suggest that when considering 
an average of estimators of the form \eqref{eq:avg-iid-unbiased}, it
is possible to sequentially refine the randomisation distribution
$(p_i)_{i\ge 1}$ used for $Z_*^{(n,j)}$ based on the random variables
generated for the previous estimators
$Z_*^{(n,1)},\ldots,Z_*^{(n,j-1)}$, as suggested by \cite{rhee-glynn}.
In particular, the probabilities $(p_i)_{i\ge 1}$ of single term
estimators could be chosen based on empirical variances of level $i$
variables generated for previous estimators. Then, $Z_*^{(n,j)}$ would
not be independent, but remain unbiased, and in fact
$(Z_*^{(n,j)}-\E Y)_{j\ge 1}$ would be martingale differences.

In some applications it is not possible to choose $(p_i)_{i\ge 1}$
which yield both finite variance and finite expected cost. Then, it is
not possible to attain a canonical square root convergence rate, but
\cite{rhee-glynn} suggest another approach: choose
$(p_i)_{i\ge 1}$ that ensure finite variance, but which imply infinite
expected cost. Using a result due to \cite{feller-inf}, they
deduce complexity results for unbiased estimators which are close to
what are possible with MLMC; see also \cite{zheng-blanchet-glynn}. However, quantifying the
efficiency in the present setting is not straightforward, so this is left for
future work.

%}}}

%%%%%%%%%%%%%%%%%%%%%%%%%%%%%%%%%%%%%%%%%%%%%%%%%%%%%%%%%%%%%%%%%%%%%%%%%%%%%%
\section{Generalised unbiased scheme and dependent randomisation}
\label{sec:stopping} %{{{

The general unbiased scheme proposed in Theorem \ref{thm:general} is
based on independent randomisation, that is, $(N_i)_{i\ge 1}$ are 
assumed independent of $(\Delta_i^{(j)})$. Such a scheme is often
appropriate in practice, but it is also possible to think of cases
where $(N_i)_{i\ge 1}$ could depend on $(\Delta_i^{(j)})$, for
instance in a stopping time fashion. It is also possible to retain
unbiasedness while replacing $\E N_i$ in the estimator by a
conditional expectation. We consider below a scheme
which accomodates both of these generalisations, while retaining
unbiasedness.

%%%%%%%%%%%%%%%%%%%%
\begin{condition} 
    \label{cond:stopping-dependence} %{{{
    Suppose $(\Delta_i^{(j)})_{i,j\ge 1}$, 
    $(Y_i)_{i\ge 0}$ and $Y$ are integrable random variables, 
    $\E Y_0=0$ and $\lim_{i\to\infty} \E Y_i = \E Y\in \R$, 
    $(\F_i)_{i\ge 0}$ are $\sigma$-algebras, and
    $(N_i)_{i\ge 1}$ are non-negative integer-valued random variables,
    such that $\E(N_i\mid \F_{i-1})\in(0,\infty)$ 
  almost surely for all
        $i\ge 1$.
\end{condition} %}}}

%%%%%%%%%%%%%%%%%%%%
\begin{condition} 
    \label{cond:stopping-dependence2} %{{{
Suppose Condition \ref{cond:stopping-dependence} holds, and
for all $i,j\ge 1$:
\begin{enumerate}[(i)]
\item \label{item:general-consistency}
  $\E\big(\Delta_i^{(j)}\charfun{j\le N_i}\;\big|\;\F_{i-1}\big)
        = (\E Y_i - \E Y_{i-1}) \P(j\le N_i\mid \F_{i-1})$
  almost surely.
\item \label{item:general-bound}
  There exists random variables $C_i$ 
        such that
        $\E\big(|\Delta_i^{(j)}|\charfun{j\le N_i}\;\big|\;\F_{i-1}\big)
        \le  C_i \P(j\le N_i\mid \F_{i-1})$
        and with $c_i \defeq \E C_i<\infty$.
\item \label{item:general-summability}
  $\sum_{i=1}^\infty c_i <\infty$ 
      and $\sum_i
      N_i<\infty$ almost surely.
\end{enumerate}
\end{condition} %}}}

%%%%%%%%%%%%%%%%%%%%
\begin{remark} 
    \label{rem:stopping} %{{{
Condition \ref{cond:stopping-dependence2} allows $(N_i)_{i\ge 1}$ to depend on 
$(\Delta_i^{(j)})_{i,j}$ in a stopping time fashion.
Namely, suppose that Condition \ref{cond:stopping-dependence} holds
and $(\Delta_i^{(j)})_{i,j}$ are as in Theorem \ref{thm:general}.
\begin{enumerate}[(i)]
    \item If $(\G_i^{(j)})_{i\ge 1,j\ge 0}$ are 
      $\sigma$-algebras such that $\F_{i-1}\subset \G_i^{(j)}$, $\{ N_i \le j \} \in \G_i^{(j)}$ 
      and $\Delta_i^{(j)}$ is independent of $\G_i^{(j-1)}$ for
      all $i,j\ge 1$, then, Condition \ref{cond:stopping-dependence2}
      \eqref{item:general-consistency} holds because
\begin{align*}
    \E(\Delta_i^{(j)} \charfun{j\le N_i}\mid \F_{i-1})
    &= \E\big[\E(\Delta_i^{(j)}\mid \G_i^{(j-1)}) \charfun{j\le
      N_i}\bigmid \F_{i-1}\big] 
    = \E \Delta_i \P( j\le N_i \mid \F_{i-1}).
\end{align*}
Condition \ref{cond:stopping-dependence2}
\eqref{item:general-bound} holds similarly, 
with $c_i = C_i = \E |\Delta_i|$.
\item In particular, if $(\Delta_i^{(j)})_{i,j}$ are independent and 
  $N_i$ are stopping times with respect to
$(\Delta_i^{(j)})$, in the sense that $\{ N_i \le j \} \in
\G_i^{(j)} \defeq 
\sigma(\Delta_k^{(1)},\ldots,\Delta_k^{(j)}\given 1\le k\le i)$
and $\F_{i-1}\subset \G_i^{(j)}$ for all $i,j\ge 1$, then
Condition \ref{cond:stopping-dependence2} holds as above.
\item If $c_i = \E|\Delta_i|$ and
  Condition \ref{cond:single-term} holds for some probability
  distribution $(p_i)_{i\ge 1}$, then 
  $\sum_{i\ge 1} c_i = \sum_{i\ge 1} c_i (\charfun{c_i
    \le p_i} + \charfun{c_i> p_i}) \le 1 + \sum_{i\ge 1} c_i^2/p_i<\infty$.
\end{enumerate}
\end{remark} %}}}

%%%%%%%%%%%%%%%%%%%%
\begin{theorem} 
    \label{thm:dependent-randomisation} %{{{
Assume Condition \ref{cond:stopping-dependence} holds and denote
\[
    Z_m \defeq \sum_{i=1}^m \frac{1}{\E(N_i\mid \F_{i-1})}
          \sum_{j=1}^\infty \Delta_i^{(j)}\charfun{j\le N_i},
\]
and $Z\defeq \lim_{m\to\infty} Z_m$ whenever well-defined.
\begin{enumerate}[(i)]
    \item 
      \label{item:stopping-gen} 
      If Condition \ref{cond:stopping-dependence2}
     \eqref{item:general-consistency} and 
      \eqref{item:general-bound} hold, then
$\E Z_m = \E Y_m$. If
      also Condition \ref{cond:stopping-dependence2}
      \eqref{item:general-summability} holds, then
      $\E Z = \E Y$.
    \item \label{item:stopping-indep}
      Suppose the assumptions of Theorem \ref{thm:general} hold
      and $(\Delta_i^{(j)})_{i,j}$ are independent of
      $(\F_i)_{i\ge 0}$. If
      \begin{align*}
    v_{\ell,m} \defeq 
    \sum_{i,k=\ell+1}^m \bigg[ &
    \cov(\Delta_i,\Delta_k)
    \E\bigg(\frac{N_i\wedge
      N_k}{\E(N_i\mid\F_{i-1})\E(N_k\mid\F_{k-1})}\bigg) \\
    &+
    \E\Delta_i\E\Delta_k\E\bigg(
      \frac{N_i N_k }{\E (N_i\mid\F_{i-1}) \E (N_k\mid\F_{k-1})} -1\bigg)\bigg]
      <\infty
     \end{align*}
     for all $0\le \ell<k<\infty$ and 
     $\lim_{k\to\infty}\sup_{j\ge 1} v_{m_k,m_{k+j}}=0$
     for some subsequence $(m_k)_{k\ge 1}$, then $\E Z = \E Y$ and 
     $\var(Z) = \lim_{k\to\infty} v_{0,m_k}$.
\end{enumerate}
\end{theorem} %}}}
Proof of Theorem \ref{thm:dependent-randomisation} is given in Appendix 
\ref{app:proof-dependent-randomisation}.

Theorem \ref{thm:dependent-randomisation} \eqref{item:stopping-indep}
is a generalisation of Theorem \ref{thm:general}, and leads to new,
potentially interesting estimators. For instance, let $N_i \defeq
\sum_{j=1}^n \charfun{R^{(j)}=i}$ where $R^{(1)},\ldots,R^{(n-1)}$ are
positive random integers independent of $R^{(n)}\sim (p_i)_{i\ge 1}$.
If we take $\F_i\defeq \sigma(R^{(1)},\ldots,R^{(n-1)})$ for all $i\ge
0$, then $\E(N_i\mid \F_{i-1}) = p_i + \sum_{j=1}^{n-1}
\charfun{R^{(j)}=i}$. This can be viewed as a residual sampling scheme
applied with the (random) probability distribution $\hat{p}_i \defeq
\E(N_i\mid \F_{i-1})/n$. Analogously, the residual sampling scheme may
be viewed through such conditioning, where $R^{(1)},\ldots,R^{(n-r)}$
are deterministic. It is unclear whether Theorem
\ref{thm:dependent-randomisation} \eqref{item:stopping-indep} allows
estimators that have practical appeal, such as greater efficiency
compared with the estimators introduced in Section \ref{sec:new}.

Theorem \ref{thm:dependent-randomisation} \eqref{item:stopping-gen}
with Remark \ref{rem:stopping} is similar to Wald's identity.
The stopping time formulation may prove theoretically useful, and 
suggests a possibility to be used in conjunction with stopping rules
developed in the MLMC context
\cite[e.g.][]{collier-haji-ali-nobile-von-scwerin-tempone}.
However, the practical relevance of such an approach may be
limited, because the expectation of a stopping time is often
unavailable. It appears also difficult to derive useful explicit 
variance expressions when $(N_i)$ depend on
$(\Delta_i^{(j)})$. The sequential refinement of $(p_i)_{i\ge 1}$
during repeated simulation of unbiased estimators, 
as discussed in Section \ref{sec:cost}, could be used instead.

%}}}

%%%%%%%%%%%%%%%%%%%%%%%%%%%%%%%%%%%%%%%%%%%%%%%%%%%%%%%%%%%%%%%%%%%%%%%%%%%%%%
\section{Numerical experiments} 
\label{sec:numerical} %{{{

The performance of the proposed estimators was studied with four SDE examples,
where expectations of final-value functionals are estimated.
Three of the models were the same as those used by
\cite{rhee-glynn}. The first is a \emph{geometric Brownian motion} (gBM)
\begin{equation*}
    \ud X_t = \mu X_t \ud t + \sigma X_t \ud B_t, \qquad
    t\in[0,1],\qquad X_0\equiv 1,
\end{equation*}
with $\mu=0.05$ and $\sigma=0.2$, where $(B_t)_{t\ge 0}$ is the
standard Brownian motion. The target functional is the 
final value European option
$f(x)=e^{-\mu}(x-1)_+$ with approximate expected value
$\E f(X_1) = 0.104505836$ \citep{rhee-glynn}.

The second model was a \emph{Cox-Ingersoll-Ross} (CIR) model
\begin{equation*}
    \ud X_t = \kappa(\theta-X_t)\ud t + \sigma \sqrt{X_t} \ud B_t,
    \qquad t\in[0,1],\qquad X_0\equiv 0.04,
\end{equation*}
with $\kappa=5$, $\theta=0.04$ and $\sigma=0.25$, and with final value
European option $f(x) = e^{-0.05}(x - 0.03)_+$. According to extensive 
simulations, the expected value was found to be approximately $\E f(X_1) =
0.01142686$.
%(contrary to \cite{rhee-glynn} who report 
%$0.0120124$).

The third model is a bivariate \emph{Heston} model
\begin{equation*}
    \begin{array}{rl}
    \ud S_t &= \mu S_t \ud t + \sqrt{V_t} S_t \ud B_t^{(1)} \\
    \ud V_t &= \kappa(\theta-V_t)\ud t + \sigma \sqrt{V_t} \ud
    B_t^{(2)},
    \end{array}
    \qquad t\in[0,1],
    \qquad
    \begin{array}{rl}
        S_0\equiv 1 \\ V_0\equiv 0,
    \end{array}
\end{equation*}
with $\mu=0.05$, $\kappa=5$, $\theta=0.04$, $\sigma=0.25$, and
where $(B_t^{(1)},B_t^{(2)})$ are coordinates of a correlated Brownian
motion with coefficient $\rho = -0.5$. The functional
$f(s,v)=f(s)=e^{-\mu}(s-1)_+$, with expected value
$\E f(S_1) = 0.10459672$ \citep{rhee-glynn,kahl-jackel}.

The fourth model is an artificial model termed \emph{modified gBM},
which has the same volatility term as 
gBM but a time-dependent drift:
\begin{equation*}
    \ud X_t = t^2X_t\ud t + \sigma X_t \ud B_t, \quad t\in[0,1],
    \qquad X_0 \equiv 1,
\end{equation*}
and $\sigma=0.1$. The target functional is the mean, which was found
to have an approximate expected value $\E X_1 = 1.395612139$. 
This last model is intended to have bigger
$|\E Y_i - \E Y|$ than the previous examples, 
highlighting the differences
between the new estimators and averages of independent estimators.
%see Remark \ref{rem:eff-diff}.

Algorithm \ref{alg:single-term} summarises an implemention of the
single-term and coupled-sum estimators $Z^{(n)}_\str$, $Z^{(n)}_\sys$,
$Z^{(n)}_{\sumsymb,\str}$ and $Z^{(n)}_{\sumsymb,\sys}$ in the SDE
context (see Section \ref{sec:mlmc}). The random variables
$(R^{(1)},\ldots,R^{(n)})$ may be constructed as in Definition
\ref{def:stratified} or Definition \ref{def:systematic}, for
$Z^{(n)}_{\uarg,\str}$ or $Z^{(n)}_{\uarg,\sys}$, respectively.
For $Z^{(n)}_{\uarg,\res}$, one may construct 
$R^{(1)},\ldots,R^{(n)}$ as in the proof of Proposition \ref{prop:consistency}
in Appendix \ref{app:proof-prop-consistency}. The independent sum
estimators may be implemented similarly as the coupled sum, 
by interchanging the lines \ref{line:bm-sim} and
\ref{for-single-coupled}.
%%%%%%%%%%%%%%%%%%%%
\begin{algorithm}
\caption{Unbiased estimators in the SDE context}
\label{alg:single-term} %{{{
\begin{algorithmic}[1]
\Function{Unbiased}{$(p_i)_{i\ge 1}$, $R^{(1)},\ldots,R^{(n)}$}
  \State $Z\gets 0$
  \For{$i=1,\ldots,\max_j R^{(j)}$}
  \State $w_i \gets 
  \begin{cases}
      (n p_i)^{-1} &\text{if single term} \\
      \big[n (1-\sum_{j=1}^{i-1} p_j)\big]^{-1}
       &\text{if coupled sum}
  \end{cases}$
  \EndFor
  \For{$j = 1,\ldots,n$} 
%  \If{$E\in\{\mathrm{Single},\mathrm{C.sum}\}$}
%  \EndIf
  \State $L \gets \begin{cases}
  \{R^{(j)}-1\}& \text{if single term} \\
  \{0,\ldots,R^{(j)}-1\}& \text{if coupled sum}
  \end{cases}$
  \State 
  \label{line:bm-sim} Simulate Brownian path
  $B$ at mesh $\tau_{R^{(j)}-1}$.
  \For{$\ell \in L $} \label{for-single-coupled}
%  \If{$E=\mathrm{I.sum}$}
%  \State Simulate Brownian path
%  $B$ at mesh $\tau_{\ell}$.
%  \EndIf
  \State \label{line:y-plus} 
  $Y_+ \gets f(X_{2^\ell}^{(\tau_{\ell})})$
  \State $Y_- \gets \begin{cases}
  0,& \ell=0 \\
  f(X_{2^{\ell-1}}^{(\tau_{\ell-1})}),&\ell>0
  \end{cases}$ 
  \State // where $X_{2^\ell}^{(\tau_{\ell})}$ and
  $X_{2^{\ell-1}}^{(\tau_{\ell-1})}$
  are based on $B$
  \State $Z \gets Z + w_{\ell+1} (Y_+ - Y_-)$
  \EndFor
  \EndFor
 \State \Return $Z$
\EndFunction
\end{algorithmic}
\end{algorithm} %}}}
The C++ source code of the implementation developed for the tests 
is available at \url{https://bitbucket.org/mvihola/unbiased-mlmc}.

In all but the Heston model, the standard Milstein scheme 
described in Section \ref{sec:mlmc} was employed. The antithetic
truncated Milstein scheme proposed by
\cite{giles-szpruch} was applied with the Heston model. 
We considered single term, independent sum and coupled sum estimators.
With the first two, 
the distributions $(p_i)_{i\ge 1}$ were set to approximately
optimal in each model as discussed in Section \ref{sec:cost}: $p_1,\ldots,p_m$ were
set empirically to minimise variance, and the tail probability
$P_\mathrm{tail}$ 
was calculated from prior simulated data.
The tail distribution $p_i$ for
$i>m$ was 
geometric, $p_i = P_\mathrm{tail} (1-2^{-\gamma})2^{-\gamma(i-m)}$,
with the parameter $\gamma$ set in all cases to $1.5$.
This provided a very good fit with the empirical optimal distribution
$\beta_i/\sum_{j\ge 1}\beta_j$ with
$m=1$ in case of gBM and $m=6$ with all other examples;
see also theoretical
results on the Milstein scheme \citep{kloeden-platen,jentzen-kloeden-neuenkirch} 
and the antithetic truncated Milstein \citep{giles-szpruch}.
In the coupled sum estimator, the differences $(\Delta_i)_{i\ge
1}$ were all based on a single Brownian trajectory. 
The optimisation was based on estimators 
of $\var(Y-Y_i)\approx \var(Y_m - Y_i)$, where $Y_m$ 
corresponded to discretisation with a mesh of $2^{13}$.
An optimal subsequence was found 
using the dynamic programming algorithm of
\cite{rhee-glynn}.

Figure \ref{fig:res-all} shows results based on averages of $10^5$
independent runs of each algorithm in each model.
%%%%%%%%%%%%%%%%%%%%
\begin{figure} %{{{
    \begin{center} 
    \includegraphics{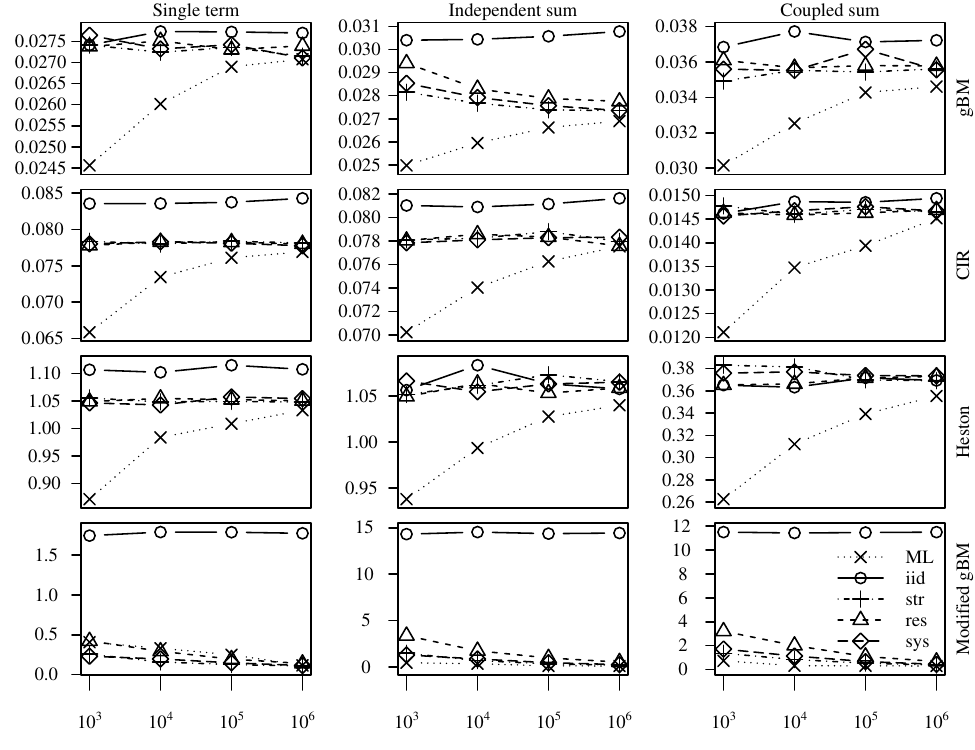}
    \end{center}
    \caption{Estimated IRE over $10^5$
      replications of estimators with $n=10^3$, $10^4$, $10^5$ and $10^6$.} 
    %}}}
   \label{fig:res-all}
\end{figure}
The graphs in Figure \ref{fig:res-all} show the estimated inverse
relative efficiency (IRE) of the methods: 
the average cost multiplied with average square
deviation from the ground truth values as given above. According to the
theory, this quantity is constant with independent averages, and with 
residual and stratified sampling, it converges to the same limit as
the corresponding MLMC estimator. Table \ref{tab:res} shows the
corresponding numerical values with $n=10^6$.
%%%%%%%%%%%%%%%%%%%%
\begin{table} %{{{
    \caption{Estimated IREs in the experiments with
      $n=10^6$.} 
\begin{center}
    \scriptsize
    \begin{tabular}{lllllllllllll}
\toprule
 & \multicolumn{3}{c}{gBM} 
 & \multicolumn{3}{c}{CIR} 
 & \multicolumn{3}{c}{Heston} 
 & \multicolumn{3}{c}{Modified gBM}  \\
 \cmidrule(lr){2-4}
 \cmidrule(lr){5-7}
 \cmidrule(lr){8-10}
 \cmidrule(lr){11-13}
 & Single & I.sum & C.sum
 & Single & I.sum & C.sum 
 & Single & I.sum & C.sum 
 & Single & I.sum & C.sum \\
\midrule
$Z_{*,\ml}^{(n)}$  & 0.0271 & 0.0269 & 0.0346 % gBM
                   & 0.0769 & 0.0775 & 0.0145 % CIR
                   & 1.0331 & 1.0401 & 0.3555 % Heston
                   & 0.1149 & 0.1182 & 0.3024 % modified gBM
                   \\
$Z_{*,\iid}^{(n)}$        & 0.0277 & 0.0308 & 0.0372 
                   & 0.0843 & 0.0816 & 0.0149 
                   & 1.1074 & 1.0579 & 0.3693 
                   & 1.7757 & 14.489 & 11.482 
                   \\
$Z_{*,\str}^{(n)}$ & 0.0271 & 0.0274 & 0.0356 
                   & 0.0782 & 0.0780 & 0.0147 
                   & 1.0479 & 1.0658 & 0.3691
                   & 0.1061 & 0.2400 & 0.3819 
                   \\
$Z_{*,\res}^{(n)}$ & 0.0274 & 0.0278 & 0.0358 
                   & 0.0780 & 0.0776 & 0.0147 
                   & 1.0503 & 1.0589 & 0.3730 
                   & 0.1368 & 0.5168 & 0.6809 
                   \\
$Z_{*,\sys}^{(n)}$ & 0.0271 & 0.0274 & 0.0355
                   & 0.0778 & 0.0783 & 0.0147
                   & 1.0543 & 1.0652 & 0.3735
                   & 0.1119 & 0.2878 & 0.4476
                   \\
\bottomrule
\end{tabular} %}}}
\end{center}
\label{tab:res} 
\end{table} 

The experiments appear to align well with the theoretical findings. In
all but the last example, the MLMC estimator admitted the best IRE
with small $n$, and the new estimators appear to admit performance in
between the independent and the MLMC case. The performance of the new
schemes come close to MLMC performance as $n$ increases, as verified
by the differences with $n=10^6$ reported in Table \ref{tab:res}, which
are all relatively small, except for the sum estimators in the
modified gBM example indicating still some discrepancy.

Note that the MLMC implemented in the experiments is the `idealised'
version involving same pre-determined allocation strategy as unbiased
estimators (Description \ref{def:mlmc}). This explains the 
discrepancy between the MLMC IRE reported here and by
\cite{rhee-glynn}, who employ the original version of the MLMC
 \citep{giles-or}. The findings of \cite{rhee-glynn}
suggest that unbiased estimators applied with stopping rules may
sometimes improve upon the original MLMC.

The differences in performance of single term and independent sum
estimators with gBM, CIR and Heston examples are all relatively small,
As $n$ increases, the IREs of the new single term and sum estimators
become negligible, as anticipated by the theory.
The coupled sum estimator appears to admit greater efficiency with CIR
and Heston examples, but the differences between independent
average estimators and the new estimators are small. The modified gBM
example demonstrates that the new estimators can be significantly more
efficient. The numerical values shown in Table \ref{tab:res} indicate
a 13--16 fold increase in terms of relative efficiency with the new
single term estimators and similar performance with MLMC. The increase
is 17--60 fold with sum estimators. In both cases, the stratified and
systematic sampling estimators appear to perform slightly better than
the residual sampling estimator.

%}}}

%%%%%%%%%%%%%%%%%%%%%%%%%%%%%%%%%%%%%%%%%%%%%%%%%%%%%%%%%%%%%%%%%%%%%%%%%%%%%%
\section{Discussion} 
\label{sec:conclusions} %{{{

This paper presented a general framework for unbiased MLMC estimation,
which admits previous debiasing schemes as special cases, and
accomodates new, lower variance unbiased estimators. The proposed
stratified sampling and residual sampling schemes are promising
classes of estimators, as they enjoy good theoretical behaviour---they
can not only improve on averages of independent estimators (Proposition
\ref{prop:consistency} and Proposition \ref{prop:cost-new}), but also
can have a significant gain in efficiency as illustrated in the
experiments. Indeed, the stratified and residual sampling estimators
can be made arbitrarily close to MLMC in efficiency under general
assumptions (Theorem \ref{thm:residual-eff} and Proposition
\ref{prop:mlmc-expected-cost}), highlighting the close connection 
between the MLMC and the debiasing schemes, and showing that 
unbiasedness may be achieved with virtually no sacrifice on efficiency.
Unbiasedness is an important quality of estimators, when 
employed as part of stochastic optimisation algorithms
\citep{borkar-sa,kushner-yin-sa}, or in a `compound sampling' context
\citep{vihola-helske-franks}. It also enables rigorous stopping rules,
which can lead to benefits over MLMC stopping rules
\citep[§4]{rhee-glynn}. 

While the theory presented in this paper does not give guarantees on
the limiting efficiency of the systematic sampling, it is expected to
behave often similar to stratified and residual sampling schemes, as
illustrated by the experiments. However, as stratified and residual
sampling enjoy good theoretical properties, and because systematic
sampling appears to perform comparatively in practice, stratified or
residual sampling are recommended as safer alternatives. The empirical
evidence from the numerical experiments in Section \ref{sec:numerical}
suggests that stratified sampling might sometimes perform slightly
better than residual sampling. This, together with the straightforward
implementation of stratified sampling, makes it appealing for
practical purposes.

Averages of independent realisations of the single term and the 
independent sum estimators may have different efficiencies in general.
In case of residual and stratified sampling, the optimally tuned
estimators often coincide in asymptotic efficiency. In general, the
single term estimator always dominates the independent sum estimator
in terms asymptotic efficiency, rendering the single term estimator
preferable over the independent sum estimator. 
Based on the experiments, the optimally tuned coupled sum estimator
may sometimes lead into significant performance gains. 
This suggests also that dependent level estimators might be worth
considering in the MLMC context, where independent level
estimators are currently widely employed. 

Despite the potential performance gain of the coupled estimators, 
it should be noted that tuning of $(p_i)_{i\ge 1}$ requires
requires estimation of $\var(Y-Y_i)\approx \var(Y_m - Y_i)$, with $m$ large, which 
is often computationally demanding.
This is in sharp contrast with  the single term estimator, where $\var(\Delta_i)$ are
easily accessible and inexpensive for moderate $i$. 
The simplicity of the single term estimator optimisation criterion
suggests, as discussed in Section \ref{sec:cost}, an algorithm which
automatically tunes the distribution $(p_i)_{i\ge 1}$ during repeated
simulation of $Z_\str^{(n)}$ or $Z_\res^{(n)}$, based on earlier
observed values for $\Delta_i^{(j)}$. 
\cite{haji-ali-nobile-von-schwerin-tempone} suggest methods for
finding optimal discretisation hierarchies in the MLMC context, which
could also be explored in the debiasing context.

The MLMC literature provides many other potential further research
topics.  It is possible to employ
essentially all techniques developed in the context of MLMC with
debiasing. These include, for instance, quasi-Monte Carlo
\citep{giles-waterhouse,dick-kuo-sloan}, adaptive time steps 
\citep{hoel-schwerin-szepessy-tempone} or adaptive importance sampling
scheme based on a drifted Brownian motion \citep{kebaier-lelong}.
It is yet unclear how well the unbiased estimators can compete
with MLMC in scenarios with slower than canonical rate of convergence.
This was investigated by \cite{rhee-glynn} and elaborated in 
\cite{zheng-blanchet-glynn}, but quantifying the effect in the context
of the new estimators requires further research. Recent work of
\cite{zheng-glynn} suggests an infinite stratification approach,
which is slightly different what was proposed here.

%}}}

%%%%%%%%%%%%%%%%%%%%%%%%%%%%%%%%%%%%%%%%%%%%%%%%%%%%%%%%%%%%%%%%%%%%%%%%%%%%%%
\section*{Acknowledgements} %{{{

This research was supported by the Academy of Finland
(grants 274740 and 284513), and computational resources were provided
by CSC, the IT Center for Science, Finland. The author wishes to thank
Christel Geiss, Stefan Geiss, Peter Glynn, Chang-Han Rhee, Raul
Tempone and Anni Toivola for useful discussions and remarks, and the
associate editor and the reviewers of \emph{Operations Research} 
for helpful comments, and
for the suggestion to consider dependent randomisation, which led to the 
developments in Section
\ref{sec:stopping}.

%}}}

%%%%%%%%%%%%%%%%%%%%%%%%%%%%%%%%%%%%%%%%%%%%%%%%%%%%%%%%%%%%%%%%%%%%%%%%%%%%%%

\appendix

%%%%%%%%%%%%%%%%%%%%%%%%%%%%%%%%%%%%%%%%%%%%%%%%%%%%%%%%%%%%%%%%%%%%%%%%%%%%%%
\section{Subsequence lemma} 
\label{app:single-sum} %{{{

The following lemma is due to an argument by \cite{rhee-glynn}.
%%%%%%%%%%%%%%%%%%%%
\begin{lemma} 
\label{lem:subsequence} %{{{
Suppose 
$\E(Y_i - Y)^2 \to 0$ as $i\to\infty$.
Then, there exists strictly
increasing  $(m_k)_{k\ge 1}$ such that for all $k\ge 1$,
\[
    \E ( Y_i - Y_{m_{k}})^2
    \le 4 
    \E( Y_i - Y)^2
    \qquad\text{for all $m_1 \le i \le m_{k}$.}
\]
\end{lemma} %}}}
\begin{proof} %{{{
Let $\delta_i \defeq Y_i - Y$.
If $\E \delta_i^2=0$ for infinitely many $i$, choose them
as $(m_k)_{k\ge 1}$. Otherwise, let $m_1$ be such that $\E
\delta_i^2>0$ for all $i\ge m_1$, and define recursively $m_{k+1}
\defeq \min\{ i>m_k\given \E \delta_i^2 \le \E\delta_{m_k}^2\}$.
Now, for any $k$ and any $m_{1}\le i\le m_k$,
\[
    \E(Y_{i}-Y_{m_k})^2 \le 2 \E\delta_i^2 + 2 \E\delta_{m_k}^2 
    \le 4 \E\delta_i^2. \qedhere
\]
\end{proof} %}}}

%}}}

%%%%%%%%%%%%%%%%%%%%%%%%%%%%%%%%%%%%%%%%%%%%%%%%%%%%%%%%%%%%%%%%%%%%%%%%%%%%%%
\section{Proof of Theorem \ref{thm:general}} 
\label{app:proof-general} %{{{
Define $Z_0\defeq 0$ and for $m\ge 1$
\[
    Z_m \defeq \sum_{i=1}^m \frac{1}{\E N_i} \sum_{j=1}^{\infty}
    \Delta_i^{(j)} \charfun{j\le N_i},
\]
then by dominated convergence
\[
    \E Z_m = \sum_{i=1}^m \frac{\E \Delta_i}{\E N_i }
    \E\bigg[\sum_{j=1}^\infty
    \charfun{j\le N_i}\bigg] = \sum_{i=1}^m \E \Delta_i
    = \E Y_m,
\]
because $\sum_{j=1}^\infty \charfun{j\le N_i} = N_i$.

For $i, k\ge 1$, by dominated convergence,
\begin{align}
\E\big[(Z_i-Z_{i-1})&(Z_k-Z_{k-1})\big]
\label{eq:single-level-diffs} \\
&= \frac{1}{\E N_i \E N_k} \E\bigg[\sum_{j,\ell=1}^\infty
\Delta_i^{(j)}\Delta_k^{(\ell)} \charfun{j\le N_i}\charfun{\ell\le N_k} 
\bigg]\nonumber\\
&= \frac{1}{\E N_i \E N_k} \bigg(\E (\Delta_i \Delta_k)
\E\bigg[\sum_{j=1}^\infty \charfun{j\le N_i}\charfun{j\le N_k}\bigg] 
\nonumber\\
&\phantom{=\frac{1}{\E N_i \E N_k}\bigg(}+ \E \Delta_i \E \Delta_k 
\E\bigg[\sum_{\stackrel{j,\ell=1}{j\neq \ell}}^\infty \charfun{j\le
  N_i}\charfun{\ell\le N_k}
\bigg]\bigg) \nonumber\\
&= \frac{\big(\E(\Delta_i\Delta_k) -
  \E\Delta_i\E\Delta_k\big)\E(N_i\wedge N_k) + \E\Delta_i\E\Delta_k\E
  [N_i N_k]}{\E N_i \E N_k},\nonumber
\end{align}
because $\sum_{j,\ell=1}^\infty \charfun{j\le N_m} \charfun{\ell\le N_k} =
N_m N_k$.

We deduce that for $0\le \ell<m$, 
\begin{align}
\E(Z_m-Z_{\ell})^2 = \sum_{i,k=\ell+1}^m 
\E\big[(Z_i-Z_{i-1})&(Z_k-Z_{k-1})\big]
= v_{\ell,m} + (\E Y_m - \E Y_\ell)^2.
\label{eq:full-level-diffs}
\end{align}
Therefore, by assumption $(Z_{m_k})_{k\ge 1}$ is Cauchy 
in $L^2$. Because $\sum_{i\ge 1}
N_i<\infty$ a.s., we have $Z_{m_k} \to Z$ a.s., and therefore $Z_{m_k}\to Z$
in $L^2$. We deduce that $\E Z = 
\lim_{k\to\infty} \E Z_{m_k} = \lim_{k\to\infty} \E Y_{m_k} = \E Y$. 
Similarly we find the expression
$\E Z^2 = \lim_{k\to\infty} \E
Z_{m_k}^2 = \lim_{k\to\infty} [v_{0,m_k} + (\E Y_{m_k})^2 ]$.
\qed

%}}}

%%%%%%%%%%%%%%%%%%%%%%%%%%%%%%%%%%%%%%%%%%%%%%%%%%%%%%%%%%%%%%%%%%%%%%%%%%%%%%
\section{Proof of Proposition \ref{prop:consistency}} 
\label{app:proof-prop-consistency} %{{{

Let us consider first $Z_\str^{(n)}$, which is an average of
independent random variables $(X_i^{(1)})_{i=1,\ldots,n}$ which
follow, respectively, 
the conditional distribution of the single term estimator $Z^{(1)}$, 
given the uniform random variable $U$ generating
$R=F^{-1}(U)$ takes
value in $I_i \defeq \big(\frac{i-1}{n},\frac{i}{n}\big)$.  The
desired variance bound follows from Lemma \ref{lem:strat}
\eqref{item:prop-alloc} applied with $(X_i^{(1)})$, $\ell_i = 1$,
$\ell = n$ and $q_i = n^{-1}$. The sum estimator
$Z_{\sumsymb,\str}^{(n)}$ is similarly stratified version of the average
of $n$ independent sum estimators. 

The systematic sampling estimators $Z_\sys^{(n)}$ and
$Z_{\sumsymb,\sys}^{(n)}$, are averages as above, but with the
difference that the uniformly distributed random variables on $I_i$
which determine $X_i^{(1)}$ are not independent. We may apply Lemma
\ref{lem:strat} \eqref{item:strat-pessimistic} which gives the
(pessimistic) upper bound on the variance.

The residual sampling $Z_\res^{(n)}$
and $Z_{\sumsymb,\res}^{(n)}$ may also be seen as
stratification of $Z^{(1)}$ and $Z_\sumsymb^{(1)}$, 
but instead of considering $R = F^{-1}(U)$, 
we now let $R = g(U)$, where
\[
    g(u) \defeq \begin{cases}
    \min\{k\in\N\given (r/n)\sum_{i\ge 1} p_i^* \ge u\},
      & u\in(0,r/n), \\
    \min\{k\in\N\given n^{-1} \sum_{i\ge 1} n_i \ge u-n/r\}, &
      u\in[r/n,1).
    \end{cases}
\]
It is direct to verify that $g(U)\sim (p_i)_{i\ge 1}$.
Let $X_1^{(1)},\ldots,X_r^{(1)}$ be independent 
random variables with conditional distribution
of $Z^{(1)}$ (resp. $Z_\sumsymb^{(1)}$) given $U\in I_i$, respectively, 
and let $X_{r+1}^{(1)},\ldots,X_{r+1}^{(n-r)}$ 
be, similarly, independent conditional on $U \in 
\big(\frac{r}{n},1\big)$. We may apply 
Lemma \ref{lem:strat}
\eqref{item:prop-alloc} with $\ell_1=\cdots=\ell_r=1$,
$\ell_{r+1}=n-r$, $\ell=n$, $q_1=\cdots=q_r = n^{-1}$ and
$q_{r+1}=(n-r)/n$.
\qed

%}}}

%%%%%%%%%%%%%%%%%%%%%%%%%%%%%%%%%%%%%%%%%%%%%%%%%%%%%%%%%%%%%%%%%%%%%%%%%%%%%%
\section{Proof of Theorem \ref{thm:residual-eff}}
\label{app:proof-eff} %{{{

Consider \eqref{item:residual-eff-single}, and
let $N_i$ correspond to $Z_\res^{(n)}$.
We have for any $i,k\ge 1$
\begin{align*}
    \cov( N_i, N_k)
    &= \cov(N_i^*, N_k^*) =
    r \big(p_i^*\charfun{i=k} -  p_i^* p_k^*\big),
\end{align*}
so $|\cov(N_i,N_k)|\le r$. In case  of $Z_\str^{(n)}$, it is not
difficult to check that the number of strata $I_j\defeq 
\big(\frac{j-1}{n},\frac{j}{n}\big)$ partially
overlapping $i$, that is, such that $i\in F^{-1}(I_j)$ but
$F^{-1}(I_j)\setminus\{i\}\neq \emptyset$, is at most two,
so $\var(N_i)\le 1$, and consequently $|\cov(N_i,N_k)|\le 1$.

Denote $Z_*^{(n,m)} \defeq n^{-1} \sum_{i=1}^m
p_i^{-1} \sum_{i=1}^{N_i} \Delta_i^{(j)}$  for $m\ge 1$,
where $(N_i)_{i\ge 1}$
correspond either to the residual sampling or stratified sampling
scheme. We have (cf.~Example \ref{ex:single})
\begin{align*}
\bigg| \var(Z_*^{(n,m)})
- \frac{1}{n}\sum_{i=1}^m \frac{\var(\Delta_i)}{p_i}\bigg|
&\le 
  \frac{r}{n^2} \sum_{i,k=1}^m 
    \frac{|\E\Delta_i\E\Delta_k| }{p_i p_k}.
\end{align*}
Because $r/n\to 0$ as $n\to\infty$,
we deduce that
$\lim_{n\to\infty} n \var(Z_*^{(n,m)})
= \sum_{i=1}^m
\var(\Delta_i)/p_i$. 

On the other hand, 
$\var(Z_*^{(n,m)}-Z_*^{(n,\ell)})$
is no greater than the corresponding variance of single term estimators
\eqref{eq:v-single},
and therefore for $1<\ell<m$,
\[
  n \var(Z_*^{(n,m)}-Z_*^{(n,\ell)})
  \le  \sum_{i=\ell+1}^m \frac{\E \Delta_i^2}{p_i} - (\E Y_m - \E
  Y_\ell)^2.
\]
We deduce that for any $m_0>1$, due to independence,
\begin{align*}
     |n \var(Z_*^{(n)})
    -\sigma_\infty^2 | 
    &=  |n \lim_{m\to\infty} \var(Z_*^{(n,m)})-\sigma_\infty^2| \\ 
    &\le   |n 
    \var(Z_*^{(n,m_0)})-\sigma_\infty^2|
    +  \limsup_{m\to\infty} 
    n \var(Z_*^{(n,m)}-Z_*^{(n,m_0)}).
\end{align*}
Both terms on the right 
can be made arbitrarily small by choosing $m_0$ large enough
and letting $n\to\infty$.

An easy calculation shows that
\begin{align}
    n\var(Z_\ml^{(n)}) 
    &= n \sum_{i=1}^{m_n} \frac{\charfun{n_i>0}}{n_i}\var(\Delta_i)
    \nonumber\\
    &= \sum_{i=1}^{m_n} \frac{\var(\Delta_i)}{p_i}
    + \sum_{i=1}^{m_n} \bigg( \frac{n\charfun{n_i>0}}{n_i} - \frac{1}{p_i}\bigg)
    \var(\Delta_i),
    \label{eq:z-ml}
\end{align}
where the first term converges to $\sigma_\infty^2$ as $n\to\infty$.
For the latter, observe that 
\[
    0\le \frac{n}{n_i} - \frac{1}{p_i} = \frac{np_i - n_i}{n_i p_i}
    \le \frac{1}{n_i p_i} \le \frac{1}{p_i},
\]
and $n_i\to\infty$ as $n\to\infty$, so by dominated convergence the
last sum in \eqref{eq:z-ml} vanishes.

Let us then turn into \eqref{item:residual-eff-sum}.
Let us calculate first
\begin{align*}
    \var(Z_{\sumsymb,\ml}^{(n)})
    &= \sum_{i,k=1}^{\tilde{m}_n}
    \frac{\cov(\Delta_i,\Delta_k)\hat{n}_{i\vee k}}{
      \hat{n}_i \hat{n}_k} \\
    &=\sum_{i=1}^{\tilde{m}_n}\bigg(
    \frac{\var(\Delta_i)}{\hat{n}_i} + 2 
    \sum_{k=i+1}^{\tilde{m}_n}
    \frac{
    \cov(\Delta_i,\Delta_k)}{\hat{n}_i}\bigg) \\
  &= \sum_{i=1}^{\tilde{m}_n}
    \frac{\var(\Delta_i)+ 2\cov(\Delta_i,Y_{\tilde{m}_n}-Y_i)}{\hat{n}_i},
\end{align*}
and because
\[
    \var(\Delta_i)+ 2\cov(\Delta_i,Y_{\tilde{m}_n}-Y_i) =
  \var(Y_{\tilde{m}_n}-Y_{i-1})-\var(Y_{\tilde{m}_n}-Y_i)\eqdef
  \xi_i^{(\tilde{m}_n)},
\]
we obtain
\begin{equation}
    n \var(Z_{\sumsymb,\ml}^{(n)})
    = \sum_{i=1}^{\tilde{m}_n}
    \frac{\xi_i^{(\tilde{m}_n)}}{\tilde{p}_i} + \sum_{i=1}^{\tilde{m}_n}\bigg(
    \frac{n}{\hat{n}_i} -
    \frac{1}{\tilde{p}_i}\bigg)\xi_i^{(\tilde{m}_n)}.
    \label{eq:mlmc-var}
\end{equation}
Note that $|\xi_i^{(\tilde{m}_n)}|\le \var(Y_{\tilde{m}_n}-Y_{i-1})
\le 
2\var(Y-Y_{\tilde{m}_n})+2\var(Y-Y_{i-1})$, so
\[  
\frac{|\xi_i^{(\tilde{m}_n)}|}{\tilde{p}_i}\le 4 \sup_{k\ge i}
\frac{\var(Y-Y_k)}{\tilde{p}_k}.
\]
The first sum in \eqref{eq:mlmc-var} therefore converges to 
$\sigma_{\sumsymb,\infty}^2$ as
$n\to\infty$ and 
the latter sum in \eqref {eq:mlmc-var} converges to zero as above.

Consider then $Z_{\sumsymb,\res}^{(n)}$. As above,
$\cov(\tilde{N}_i,\tilde{N}_k)
    =\cov(\tilde{N}^*_i, \tilde{N}^*_k) = 
    r (\tilde{p}_{i\vee k}^*
    - \tilde{p}_i^*\tilde{p}_k^*)$.
In case of $Z_{\sumsymb,\str}^{(n)}$,
at most one stratum $F^{-1}(I_j)$ contains some $\ell<i$ and $i$,
so $\var(\tilde{N}_i)\le 1/4$. Therefore,
$|\cov(\tilde{N}_i,\tilde{N}_k)|\le r$ for both schemes,
and the variance of 
$Z_*^{(n,m)}
\defeq n^{-1} \sum_{i=1}^m \tilde{p}_i^{-1} \sum_{j=1}^{\tilde{N}_i} 
\Delta_i^{(j)}$ admits the bound
\begin{align*}
\bigg|    \var(Z_*^{(n,m)})
- \frac{1}{n}\sum_{i=1}^m \bigg(
    \frac{\var(\Delta_i)}{\tilde{p}_i}
    + 2 \sum_{k=i+1}^m &\frac{\cov(\Delta_i,\Delta_k)}{\tilde{p}_i}
    \bigg)
    \bigg|
    \le 
    \frac{r}{n^2}\sum_{i,j=1}^m \frac{|\E \Delta_i \E\Delta_k|}{
      \tilde{p}_i\tilde{p}_k} 
\end{align*}
The second term on the left equals 
$n^{-1} \sum_{i=1}^m \xi_i^{(m)}/\tilde{p}_i$, and
clearly $\var(Z_{\sumsymb,\res}^{(n,m)})
\to  \sum_{i=1}^m
\xi_i^{(m)}/\tilde{p}_i$ as $n\to\infty$. Now, take $(m_k)_{k\ge 1}$
from Lemma \ref{lem:subsequence}, then stratification implies that
for $k,j\ge 1$,
\[
    n \var(\tilde{Z}_*^{(n,m_{k+j})}
    - \tilde{Z}_*^{(n,m_{k})})
    \le 4 \sum_{i=m_k+1}^{m_{k+j}} \frac{\E(Y_{i-1}-Y_{m_k})^2}{\tilde{p}_i}
    + (\E Y_{m_{k+j}}-\E Y_{m_k})^2.
\]
We conclude as above by writing for any $k\ge 1$,
\begin{align*}
|n \var(\tilde{Z}_*^{(n)}) -
\sigma_{\sumsymb,\infty}^2| & \le |n
\var(\tilde{Z}_*^{(n,m_{k})})-\sigma_{\sumsymb,\infty}^2|
+ \limsup_{j\to\infty} 
 n \big|\var (\tilde{Z}_*^{(n,m_{k+j})}
-\tilde{Z}_*^{(n,m_{k})}) \\
&\phantom{\le |n
\var(\tilde{Z}_*^{(n,m_{k})})- }
  + 2  \cov\big(
(\tilde{Z}_*^{(n,m_{k+j})}
-\tilde{Z}_*^{(n,m_{k})}),
\tilde{Z}_*^{(n,m_{k})})\big)\big|.
\end{align*}
Because $n |\cov(A,B)|\le \sqrt{n\var(A)n\var(B)}$, all terms on the
right can be made arbitrarily small by 
by choosing
$m_k$ large enough and letting $n\to\infty$.

Finally, the proof of 
\eqref{item:eff-independent-sum} follows similarly as
\eqref{item:residual-eff-single} 
and  \eqref{item:residual-eff-sum}.
\qed

%}}}

%%%%%%%%%%%%%%%%%%%%%%%%%%%%%%%%%%%%%%%%%%%%%%%%%%%%%%%%%%%%%%%%%%%%%%%%%%%%%%
\section{Proof of Theorem \ref{thm:dependent-randomisation}}
\label{app:proof-dependent-randomisation} %{{{
Consider \eqref{item:stopping-gen} and notice that
\begin{align*}
    \E\bigg(\sum_{j=1}^n |\Delta_i^{(j)}|
    \charfun{j\le N_i}\;\bigg|\;\F_{i-1}\bigg)
    &\le C_i \sum_{j=1}^n \P(j\le N_i\mid \F_{i-1}) 
    \le C_i \E(N_i\mid \F_{i-1}).
\end{align*}
Denote $S_i\defeq \E(N_i\mid \F_{i-1})^{-1} \sum_{j\ge 1}
\Delta_i^{(j)}\charfun{j\le N_i}$, then
$\E|S_i|\le c_i$ and
by dominated convergence a similar calculation yields
\[ 
    \E\bigg(
  \frac{S_i}{\E(N_i\mid \F_{i-1})}\;\bigg|\; \F_{i-1}\bigg) = 
  \lim_{n\to\infty} \E\bigg(\sum_{j=1}^n
  \frac{(\E Y_i - \E Y_{i-1}) \P(j\le N_i\mid \F_{i-1})}{\E(N_i\mid
    \F_{i-1})}\bigg) = \E Y_i - \E Y_{i-1},
\] 
which leads to $\E Z_m = \E Y_m$. 
For the latter claim, note that only finitely many
of $\{N_i\}_{i\ge 1}$ are non-zero, so $Z$ is well-defined,
and the result
follows by the dominated
convergence theorem.

The statement \eqref{item:stopping-indep} is a generalisation of
Theorem \ref{thm:general}, and the proof follows similarly. Indeed,
Condition \ref{cond:stopping-dependence2}
\eqref{item:general-consistency} 
and \eqref{item:general-bound} hold by independence, with
$c_i = C_i = \E |\Delta_i|$, so part
\eqref{item:stopping-gen} implies $\E Z_m = \E Y_m$. 
Denoting $\F_{i,k} \defeq
\sigma(\F_{i-1},\F_{k-1})$, a straightforward calculation similar
to \eqref{eq:single-level-diffs} yields
\begin{align*}
    \E\big[(Z_i-Z_{i-1})(Z_k-Z_{k-1})\bigmid \F_{i,k}\big]
    = \frac{\cov(\Delta_i,\Delta_k)\E(N_i\wedge N_k\mid \F_{i,k})
      + \E\Delta_i\E\Delta_k \E(N_i N_k\mid \F_{i,k})}{\E(N_i\mid
      \F_{i-1})\E(N_k\mid \F_{k-1})},
\end{align*}
and as in \eqref{eq:full-level-diffs}, 
$\E(Z_m-Z_{\ell})^2 
= v_{\ell,m} + (\E Y_m - \E Y_\ell)^2$.
\qed 
%}}}

\end{document}